\documentclass[10pt,reqno]{amsart}
\usepackage[a4paper,top=3cm, bottom=2.8cm, left=3cm, right=3cm, twoside]{geometry}
\usepackage{amsthm}
 \usepackage{amsmath}
 \usepackage{amsfonts}
 \usepackage{amssymb}
 \usepackage{amscd}
 \usepackage{url}
 \usepackage[utf8x]{inputenc}
\usepackage{csquotes}
 \usepackage[T1]{fontenc}
\usepackage[sfdefault]{FiraSans}
 \usepackage[utf8x]{inputenc}
 \usepackage[english]{babel}
 \usepackage[mathscr]{eucal}
 \usepackage{lmodern}
 \usepackage{bera}
 \usepackage{float} 
 \usepackage{bold-extra}
 \usepackage{textcomp}
 \usepackage{graphicx}
 \usepackage{caption}
\usepackage{subcaption} 
 \usepackage[dvipsnames,table]{xcolor}
 \usepackage{mathtools} 
 \usepackage{enumitem}
 \usepackage{tikz} 
 \usepackage{tikz-cd}
 \usepackage{wasysym}
 \usepackage{todonotes}
 \usepackage{fancybox}
 \usepackage{fancyhdr}
 \usepackage{xspace}
\usepackage[nodisplayskipstretch]{setspace}
\usepackage{hyperref}
\usepackage[maxbibnames=99,
 style=numeric,
    sorting=nyt
]{biblatex}
\DeclareFieldFormat{doi}{}
\DeclareFieldFormat{url}{}
\DeclareFieldFormat{issn}{}
\DeclareFieldFormat{isbn}{}
\DeclareFieldFormat{note}{}

\DeclareDelimFormat{nametitledelim}{\addspace}
\expandafter\let\csname ver@natbib.sty\endcsname\relax

 \theoremstyle{definition}  
  \newtheorem{defn}{Definition}[section]
  \newtheorem{eg}[defn]{Example}

   \newtheorem{rmk}[defn]{Remark}

  \theoremstyle{plain}  
  \newtheorem{thm}[defn]{Theorem}
  \newtheorem{lem}[defn]{Lemma}
  \newtheorem{prop}[defn]{Proposition}
  \newtheorem{cor}[defn]{Corollary}
  
  \newtheorem{que}[defn]{Question}

  \theoremstyle{remark}

 \renewcommand{\it}[1]{\textit{#1}}
 
 \renewcommand{\sf}[1]{\textsf{#1}}

 \newcommand{\mbb}[1]{\mathbb{#1}}
 \newcommand{\mcl}[1]{\mathcal{#1}}

 \newcommand{\mbf}[1]{\mathbf{#1}}
 \newcommand{\msc}[1]{\mathscr{#1}}

 \newcommand{\ol}[1]{\overline{#1}}

 \newcommand{\abs}[1]{\left\lvert#1\right\rvert}

 \newcommand{\norm}[1]{\left\lVert#1\right\rVert}

 \newcommand{\M}[1]{\mbb{M}_{#1}}

 \newcommand{\B}[1]{\msc{B}({#1})}

 \newcommand{\ranko}[2]{|{#1}\rangle\langle{#2}|}

 \newcommand{\ip}[1]{\left\langle#1\right\rangle}

 \newcommand{\Matrix}[1]{\begin{bmatrix}#1\end{bmatrix}}

 \DeclareMathOperator{\T}{\sf{T}}
 \DeclareMathOperator{\tr}{\sf{tr}}
 
 \DeclareMathOperator{\id}{\sf{id}}

 \DeclareMathOperator{\lspan}{\sf{span}}

 \numberwithin{equation}{section}
 \allowdisplaybreaks[4] 
 \setlist[enumerate]{font=\upshape,noitemsep, topsep=0pt} 
 \setlist[itemize]{noitemsep, topsep=0pt}

\makeatletter
\@namedef{subjclassname@2020}{\textup{2020} Mathematics Subject Classification}
\makeatother
\title{Noncommutative Marcus-Ree inequality and Erd\H{o}s channels}

\author{Repana Devendra}
\address{Department of Mathematics, Indian Institute of Technology, Bombay}
\email{r.deva1992@gmail.com, rdevamaths@gmail.com}

\author{Chaman Kumar Sahu}
\address{Department of Mathematics, Indian Institute of Technology, Bombay}
\email{sahuchaman9@gmail.com, chamanks@math.iitb.ac.in}

\begin{document}

\date{\today}
 \begin{abstract}
Unital quantum channels are the noncommutative analogues of doubly stochastic matrices. In
1959, Marcus and Ree showed that the squared Frobenius norm of a doubly stochastic matrix
is at most its maximal diagonal sum. We prove a quantum version of this inequality for
mixed unitary channels on $\mathbb{M}_d$ and give an example showing that it fails for
general unital channels. As in the classical case, we call the channels attaining equality
as \emph{Erd\H{o}s channels}. 

We show that the trace-preserving conditional expectation onto any unital $C^*$-subalgebra
of $\mathbb{M}_d$ is an Erd\H{o}s channel. It is also locally isolated among channels of
the same Choi rank. Next, we introduce quantum RCDS channels. These are the noncommutative
analogue of doubly stochastic matrices whose supported permutations all have the same
diagonal sum. We prove a structure theorem for them and use it to characterize the quantum
RCDS Erd\H{o}s channels. We give an explicit form, up to unitary equivalence, for the Erd\H{o}s channels of Choi
rank $2$, and further show that there are uncountably many unitary equivalence
classes of Erd\H{o}s channels unlike the classical case. We also characterize the Erd\H{o}s channels among mixed Weyl-unitary channels of
any Choi rank. 

For qubits, we prove that there are exactly four unitary equivalence classes of Erd\H{o}s
channels and they are represented by the identity channel and uniform averages of Pauli
channels. Finally, restricting to the diagonal channels, our results recover the Marcus--Ree inequality,
the structure of RCDS doubly stochastic matrices, and the characterization of Erd\H{o}s
matrices.
\end{abstract}



\maketitle
\section{Introduction} 
\subsection{Motivation.}
Marcus and Ree published a seminal paper \cite{MR59} aimed at understanding the van der Waerden permanent conjecture, which was eventually settled independently by Egorychev \cite{Eg81} and Falikman \cite{Fa81} in 1981. It was Erd\H{o}s who, while discussing the problem with Marcus and Ree, observed that if the van der Waerden conjecture is true, then every doubly stochastic matrix $A=[a_{jk}]_{j,k=1}^d \in\mathbb{M}_d$ (i.e., $a_{jk} \geq 0$ and $\sum_j a_{jk} = \sum_{k} a_{jk} = 1$ for all $j,k=1, \ldots, d$) must satisfy $\sum_{j=1}^d a_{j,\sigma(j)} \geq 1$, for some permutation  $\sigma \in \mathbf{S}_d$. Marcus and Ree considerably strengthened this inequality by showing that there exists a permutation $\sigma \in \mathbf{S}_d$, the set of all permutations of $\{1, 2, \ldots, d\}$, such that $ \sum_{j,k=1}^d a_{jk}^2 \leq \sum_{j=1}^d a_{j \sigma(j)}.$ This is equivalent to
\begin{equation}\label{eq-M-R-Ineq}
    \sum_{j,k=1}^d a_{jk}^2 \leq  \max_{\sigma \in \mathbf{S}_d} \sum_{j=1}^d a_{j \sigma(j)},
\end{equation}
which is now known as the Marcus--Ree inequality. This inequality prompted Erd\H{o}s to pose the following question: what are all the doubly stochastic matrices $[a_{jk}]_{j,k=1}^d \in \mathbb{M}_d$ for which equality holds in \eqref{eq-M-R-Ineq}? Note that the matrix algebra $\mathbb{M}_d$ can be viewed as a Hilbert space with the Hilbert--Schmidt inner product $\ip{A,B}=\tr(A^*B)$ for all $A,B\in\mathbb{M}_d$, and that the sum $\sum_{j=1}^d a_{j \sigma(j)}$ can be realized as $\tr(P_{\sigma}A)$, where $P_{\sigma}$ is the permutation matrix given by $P_{\sigma}(e_j)=e_{\sigma(j)}$ for all $j =1, \ldots, d$. With this viewpoint, the inequality \eqref{eq-M-R-Ineq} may be interpreted as $\ip{A,A}\leq \max_{\sigma \in \mathbf{S}_d}\ip{P_\sigma,A}$. Consequently, the question of Erd\H{o}s may be reformulated as follows: what are all the doubly stochastic matrices $A \in\mathbb{M}_d$ such that
\begin{equation*}
\ip{A,A}=\max_{P\in\mathcal{P}_d}\ip{P,A}\, ?
\end{equation*}
Here, $\mathcal{P}_d$ denotes the set of all $d\times d$ permutation matrices.

A doubly stochastic matrix $A=[a_{jk}]_{j,k=1}^d$ is called an \emph{Erd\H{o}s matrix} if it satisfies $\ip{A,A}=\max_{P\in\mathcal{P}_d}\ip{P,A}$. Observe that if $A$ is an Erd\H{o}s matrix, then so is $PAQ$ for any $P,Q\in\mathcal{P}_d$. Two doubly stochastic matrices $A$ and $B$ are said to be \emph{equivalent} if $A=PBQ$ for some permutation matrices $P, Q\in\mathcal{P}_d$. Thus, to address the question above, it suffices to classify all Erd\H{o}s doubly stochastic matrices up to equivalence. This question was answered in \ \cite{LuJa24} for the cases $d=2, 3$, where the authors proved that there are, respectively, only $3$ and $6$ doubly stochastic matrices up to equivalence.  The case $d=4$ was subsequently studied in 2025; see \cite{AmRa25}, and the cases $d=5,6$ were studied in 2026; see \cite{PHSK26}. They showed that there are exactly $41$ Erd\H{o}s doubly stochastic matrices up to equivalence for $d=4$, and at most $\frac{2^{d^2}}{2(d!)^2}$ equivalence classes for $d=5,6$. Despite this progress, a complete classification of Erd\H{o}s matrices for general $d$ remains elusive. The rapid growth in the number of equivalence classes, together with the combinatorial complexity already visible for $d=4,5,6$, suggests that a closed-form or algorithmic classification for arbitrary $d$ is unlikely to follow directly from these low-dimensional techniques. We refer the reader to \cite{PHSK26,Rag25} and the references therein for further developments and related open problems concerning Erd\H{o}s matrices.

\subsection{Noncommutative setup.}
Unital completely positive maps on finite von Neumann algebras $(\mathcal{M},\tau)$, where $\tau$ is a faithful normal state, are the natural quantum analogue of stochastic matrices: if $\mathcal{M}$ is a finite-dimensional commutative von Neumann algebra (i.e., $\mathcal{M} \cong \mathbb{C}^n$ for some $n \in \mathbb{N}$), a unital completely positive map $\Phi$ on $\mathcal{M}$ corresponds to a matrix with non-negative entries in which each row sum equals $1$. If, in addition, $\Phi$ is trace-preserving, i.e., $\sum_{i=1}^n \Phi(x)_i = \sum_{i=1}^n x_i$ for all $x = (x_i)_{i=1}^n \in \mathbb{C}^n$ (equivalently, the dual map $\Phi^*$ satisfies $\Phi^*(1) = 1$), then $\Phi$ corresponds to a matrix whose column sums are also equal to $1$. Thus, a unital, trace-preserving completely positive  map (also known as \emph{unital quantum channel})  on $\mathbb{C}^n$ corresponds to a doubly stochastic matrix of order $n$. This motivates the study of unital, trace-preserving completely positive maps on von Neumann algebras. In this article, we restrict ourselves to the von Neumann algebra $\mathbb{M}_d$.

The Marcus--Ree inequality \eqref{eq-M-R-Ineq} relies heavily on the classical Birkhoff--von Neumann theorem \cite{Bir46}, which states that the extreme points of the compact convex set of doubly stochastic matrices are exactly the permutation matrices, the \emph{reversible} ones. Consequently, by the Krein--Milman theorem, every doubly stochastic matrix can be expressed as a convex combination of permutation matrices. The reversible unital quantum channels on $\mathbb{M}_d$ are the unitary conjugations $\mathrm{Ad}_U(X):=U^*XU$, where $U$ is a unitary. However, the analogue of Birkhoff's theorem fails once $d\ge 3$, since there exist unital quantum channels that are not of this form \cite{KuMa87,Tre86}, although it holds for $d=2$ \cite{BPS93,LaSt93}.
A unital quantum channel lying in the convex hull of all unitary conjugations is called a \emph{mixed unitary} channel (see Definition~\ref{mixed-unitary-defn}).
They provide a rich class of examples that capture key properties of general quantum operations \cite{Ros08}. By virtue of their structural simplicity, they remain highly useful for theoretical analysis, often proving more tractable than arbitrary quantum channels. Such channels arise naturally in several contexts, including majorization frameworks based on Hermitian operators \cite{AU82}, as well as in quantum cryptographic protocols involving state encryption \cite{AMTB00, AS04, HDPW04}. For further details on mixed unitary channels, see \cite{BPS93, LaSt93, MeWo09, Wat18}.
\subsection{Our contribution in this paper}
We begin by establishing a quantum analogue of the Marcus--Ree inequality for mixed unitary channels on $\mathbb{M}_d$, namely
$$
\ip{\Phi, \Phi} \leq \max_{U \in \mathcal{U}(d)} \langle \mathrm{Ad}_U, \Phi \rangle \quad \text{for every mixed unitary channel } \Phi.
$$ Here, $\mathcal{U}(d)$ denotes the group of all unitary matrices of order $d$.
Example~\ref{failure-MR-ineq} shows that this inequality cannot be extended to the set of all unital quantum channels.
For a matrix $A=[a_{ij}]\in\M{d}$, let $\Phi_A$ be the \emph{diagonal} linear map on $\M{d}$ defined by
\begin{equation}\label{phiA}
\Phi_A(X): =  \sum_{i,j=1}^{d}x_{jj}a_{ij}E_{ii}, \quad X=[x_{ij}]\in \M{d}.
\end{equation}
Interestingly, it turns out that $A$ is a doubly stochastic matrix if and only if $\Phi_A$ is a mixed unitary channel (Proposition~\ref{prop-doubly-mixed}). Employing on this channel associated to a doubly stochastic matrix, we then recover the classical Marcus-Ree inequality from the quantum analog above. Furthermore, as in the classical case, we show that the difference
$$\Delta_d(\Phi) := \max_{U \in \mathcal{U}(d)} \langle \mathrm{Ad}_U, \Phi \rangle - \langle \Phi, \Phi \rangle, \qquad \Phi \in \{\text{mixed unitary channels}\},$$
cannot be arbitrarily large; more precisely, the range of $\Delta_d$ over the set of all mixed unitary channels is exactly $\left[0, \frac{d^2-1}{4}\right]$. A natural question then is to determine which mixed unitary channels  $\Delta_d$ maps to the boundary points. We characterize all the mixed unitary channels which maps to $\frac{d^2-1}{4}$ (Proposition~\ref{prop-Delta-range}). Now, it is left to classify all the mixed unitary channels $\Phi$ satisfying $\Delta_d(\Phi) = 0$. We call such channels as \emph{Erd\H{o}s channels} (Definition~\ref{def-Erdos-channel}).
Since there is a natural unitary equivalence relation between unital quantum channels (Definition~\ref{equivalence-reln}), we pose the following question.
\begin{que}
Classify all Erd\H{o}s channels on $\mathbb{M}_d$ up to unitary equivalence.
\end{que}
A prominent example of an Erd\H{o}s channel (see Proposition~\ref{prop-Delta-range}) is the well-known completely depolarizing channel , given by
$$\delta_d(X) = \frac{\tr(X)\,I_d}{d}, \quad X \in \mathbb{M}_d.$$
The channel $\delta_d$ can be viewed as the unique trace-preserving conditional expectation onto the unital $C^*$-subalgebra $\mathbb{C}I_d$ of $\mathbb{M}_d$ (see Definition~\ref{conditional-exp} for the notion of conditional expectation). Given a unital $C^*$-subalgebra $\mathcal{A}$ of $\mathbb{M}_d$, there is a unique trace-preserving conditional expectation onto $\mathcal{A}$, denoted by $\mathbf{E}_{\mathcal{A}}$. It is then natural to raise the following question.
\begin{que}
    Is the trace-preserving conditional expectation $\mathbf{E}_{\mathcal{A}}$ an Erd\H{o}s channel?
\end{que}
We obtain an affirmative answer to this question (Theorem~\ref{cond-exp-erdos}). Furthermore, we establish a refined version of the noncommutative Marcus--Ree inequality for any mixed unitary channel that fixes every element of some unital $C^*$-subalgebra $\mathcal{A}$ of $\mathbb{M}_d$, and we show that the Choi-rank of $\mathbf{E}_{\mathcal{A}}$ equals $\dim(\mathcal{A}')$, where $\mathcal{A}'$ denotes the commutant of $\mathcal{A}$ in $\mathbb{M}_d$. Building on these two facts, we prove that $\mathbf{E}_{\mathcal{A}}$ is locally isolated within the class of Erd\H{o}s channels of Choi-rank $\dim(\mathcal{A}')$ that fix every element of $\mathcal{A}$ (Theorem~\ref{nonErdos-suff-cond}).

Since the Choi-rank of $\mathbf{E}_{\mathcal{A}}$ is $\dim(\mathcal{A}')$, and not every integer $k$ with $1 \leq k \leq d^2$ arises as $\dim(\mathcal{A}')$ for some unital $C^*$-subalgebra $\mathcal{A}$ of $\mathbb{M}_d$ -- for instance, the possible dimensions of unital $C^*$-subalgebras of $\mathbb{M}_2$ are $\{1,2,4\}$ -- the following question arises.
\begin{que}\label{all rank -question}
    Given $d \in \mathbb{N}$ and an integer $k$ with $1 \leq k \leq d^2$, does there exist an Erd\H{o}s channel on $\mathbb{M}_d$ whose Choi-rank equals $k$?
\end{que}
To address this question, we introduce a quantum analog of RCDS  doubly stochastic matrices to the unital quantum channels on $\M{d}$, and obtain a structural characterization; see Definition~\ref{quantum-rcds-def} and Theorem~\ref{thm-Strong qRCDS-struc}. As a consequence, we classify Erd\H{o}s channels among all mixed Weyl-unitary channels (Theorem~\ref{Erdos-d}), and Question~\ref{all rank -question} is answered affirmatively in Corollary~\ref{Erdos-d-weyl}. We also prove that, $A$ is a RCDS matrix if and only if the channel $\Phi_A$, defined by \eqref{phiA}, is RCDS (Theorem~\ref{thm-RCDS-quaRCDS}).

Although a complete classification of all Erd\H{o}s channels up to unitary equivalence appears to be challenging, much as in the classical setting. We succeed in classifying all RCDS mixed unitary channels that are Erd\H{o}s (Theorem~\ref{thm-RCDS-Erdos}). As an application, we recover the classical characterization of RCDS Erd\H{o}s doubly stochastic matrices (Corollary~\ref{classical-Erdos}). Furthermore, we concretely classify all Erd\H{o}s channels of Choi-rank $2$ (Theorem~\ref{Rank2-classi}). In a strict contrast to the classical case, it turns out that there are uncountably many inequivalent classes of Erd\H{o}s channels of Choi-rank $2$ (Remark~\ref{inequivalent-channels-rank2}).
We also classify all Erd\H{o}s channels on $\mathbb{M}_2$ up to unitary equivalence, showing that there are exactly four such equivalence classes (Theorem~\ref{Erdos-2}).

\subsection{Organization of the paper.} The present paper is organized as follows. In the next section, we recall the notion of unital quantum channels, including mixed unitary channels, and review some known facts needed for our purposes. In Section~\ref{S3}, we establish a noncommutative Marcus-Ree inequality and introduce the notion of Erd\H{o}s channels. In Section~\ref{S4}, we prove that the trace-preserving conditional expectation is an Erd\H{o}s channel. In Section~\ref{S5}, quantum RCDS channels are characterised, and Erd\H{o}s channels among them are also found. Further, all Erd\H{o}s channels of Choi-rank $2$ are classified, and a complete classification of Erd\H{o}s channels on $\M{2}$ is also obtained.
\section{Basic theory of quantum channels}
\subsection{Notation}
Let $\mathbb{C}^d$ denote the $d$-dimensional complex Hilbert space, equipped with an inner product that is linear in the second variable and conjugate linear in the first, and let $\{e_i: 1\leq i\leq d\}$ be its standard orthonormal basis. We denote the space of $d_1\times d_2$ complex matrices by $\mathbb{M}_{d_1\times d_2}$, and abbreviate $\mathbb{M}_{d\times d}$ to $\mathbb{M}_d$. The identity matrix in $\mathbb{M}_d$ is denoted $I=I_d$, the trace by $\tr=\tr_{d}: \mathbb{M}_d\to \mathbb{C}$, and the identity map by $\id=\id_d:\mathbb{M}_d\to\mathbb{M}_d$. For a matrix $A = [a_{ij}] \in \mathbb{M}_{d_1 \times d_2}$, its adjoint $A^* \in \mathbb{M}_{d_2 \times d_1}$ is the matrix whose $(i,j)$-th entry is $\overline{a_{j,i}}$. We denote the set of positive semi-definite matrices in $\mathbb{M}_d$ by $\mathbb{M}_d^+$, writing $A\geq 0$ when $A\in\mathbb{M}_d^+$, and the group of $d\times d$ unitary matrices by $\mathcal{U}(d)$. For $x\in\mathbb{C}^{d_1}$ and $y\in\mathbb{C}^{d_2}$, the rank-one linear map $\ranko{x}{y}:\mathbb{C}^{d_2}\to\mathbb{C}^{d_1}$ is defined by $\ranko{x}{y}(z):=\ip{y,z}x$ for all $z\in\mathbb{C}^{d_2}$. For $A\in\mathbb{M}_{d_1}$ and $B\in\mathbb{M}_{d_2}$ one has $A\ranko{x}{y}B=\ranko{Ax}{B^*y}$, and $\ranko{x}{x}$ is positive for every $x\in\mathbb{C}^{d}$. For standard basis vectors $\{e_k\}$, let $E_{ij} := |e_i\rangle\langle e_j|$ denote the matrix unit, or equivalently the matrix with $1$ in the $(i,j)$ entry and $0$ elsewhere.

A linear map $\Phi:\mathbb{M}_{d_1}\to\mathbb{M}_{d_2}$ is \emph{completely positive} (a \emph{CP map}) if, for every $k\geq 1$, the map $\id_k\otimes\Phi:\mathbb{M}_{k}\otimes\mathbb{M}_{d_1}\to\mathbb{M}_{k}\otimes\mathbb{M}_{d_2}$ is positive, i.e., $(\id_k\otimes\Phi)\big((\mathbb{M}_{k}\otimes\mathbb{M}_{d_1})^+\big)\subseteq(\mathbb{M}_{k}\otimes\mathbb{M}_{d_2})^+$. By Choi's theorem \cite{Cho75}, $\Phi$ is completely positive if and only if its \emph{Choi matrix}
  \begin{align*}
        C_\Phi:=\sum_{i,j=1}^{d_1}E_{ij}\otimes\Phi(E_{ij})=\big[\Phi(E_{ij})\big]_{i,j=1}^{d_1}\in\mathbb{M}_{d_1}\otimes\mathbb{M}_{d_2}
   \end{align*}
is positive semi-definite, where $E_{ij}:=\ranko{e_i}{e_j}\in\mathbb{M}_{d_1}$ for $1\leq i,j\leq d_1$. Equivalently, $\Phi$ admits a \emph{Kraus decomposition} \cite{Kra71},
        $$\Phi=\sum_{i=1}^n\mathrm{Ad}_{V_i}, \qquad \mathrm{Ad}_V(X):=V^*XV,$$
for some \emph{Kraus operators} $V_i\in\mathbb{M}_{d_1\times d_2}$. A Kraus decomposition is not unique; however, if $\Phi=\sum_{i=1}^n\mathrm{Ad}_{V_i}=\sum_{i=1}^m\mathrm{Ad}_{W_i}$ are two such decompositions, then $\lspan\{V_i:1\leq i\leq n\}=\lspan\{W_i: 1\leq i\leq m\}$. The minimal number of Kraus operators needed to represent $\Phi$ is called the \emph{Choi-rank} of $\Phi$; it equals the rank of the Choi matrix $C_\Phi$.
\begin{defn}\label{mixed-unitary-defn}
 A unital quantum channel $\Phi:\mathbb{M}_d \to \mathbb{M}_d$ is said to be \emph{mixed unitary} if there exist $\{\lambda_j\}_{j=1}^{\ell}\subseteq [0,1]$ and unitaries $\{U_j\}_{j=1}^\ell\subseteq\mathcal{U}(d)$ with $\sum_{j=1}^\ell\lambda_j=1$ such that
 \begin{align*}
     \Phi(X)=\sum_{j=1}^{\ell}\lambda_j U_j^{*}X U_j, \quad X\in \mathbb{M}_d.
 \end{align*}
\end{defn}
Given $d \in \mathbb{N}$, let $\mathcal{MU}_d$ denote the set of all mixed unitary channels on $\mathbb{M}_d$; that is,
\begin{equation*}
    \mathcal{MU}_d := \Big\{ \sum_{j=1}^\ell \lambda_j \mathrm{Ad}_{U_j} : \ell \in \mathbb{N},~ \{U_j\}_{j=1}^\ell \subseteq \mathcal{U}(d),~ 0 \leq\lambda_j \leq 1,~\sum_{j=1}^\ell \lambda_j =1\Big\}.
\end{equation*}
The set $\mathcal{MU}_d$ is clearly convex, and since the unitary group $\mathcal{U}(d)$ is compact, $\mathcal{MU}_d$ is a compact convex subset of the set of all unital quantum channels.
The \emph{mixed unitary rank} of $\Phi$ is the minimal $\ell$ for which such a decomposition exists. It is immediate that the Choi-rank of $\Phi$ is at most its mixed unitary rank. In \cite{Bus06} it is proved that the mixed unitary rank of $\Phi$ is at most the square of its Choi-rank. For more details on the mixed unitary rank, see \cite{CKLi22}.
\section{noncommutative Marcus-Ree Inequality and Erd\H{o}s channels}\label{S3}
Motivated by the classical study of the Marcus-Ree inequality for doubly stochastic matrices, it is intriguing to investigate if there is a quantum or noncommutative analog of the Marcus-Ree inequality for unital quantum channels on $\M{d}$. 
Recall that by viewing $\M{d}$ as a Hilbert space with Hilbert-Schmidt inner product, the Marcus-Ree inequality \eqref{eq-M-R-Ineq} is equivalent to
\begin{align*}
    \ip{A,A}\leq\max_{P\in\mcl{P}_d}\ip{P,A}.
\end{align*}
This description of the Marcus--Ree inequality allows us to extend to the noncommutative setup. We equip $\B{\mathbb{M}_d,\mathbb{M}_d}$,  the space of all linear operators on $\M{d}$ with the inner product
\begin{equation*}
\ip{ \Phi, \Psi} := \tr(C_\Phi^* C_\Psi), \qquad \Phi, \Psi \in  \B{\mathbb{M}_d, \mathbb{M}_d},
\end{equation*}
where $C_\Phi$ denotes the Choi matrix of $\Phi$, turning $\B{\mathbb{M}_d,\mathbb{M}_d}$ into an inner product space. Using $E_{ij}E_{kl} = \delta_{j,k}E_{il}$ and $\tr\big(E_{ij}^*E_{kl}\big)=\delta_{i,k}\delta_{j,l}$, we obtain
\begin{equation}\label{inner-pdt-compute}
\ip{\Phi,\Psi} = \sum_{i,j,k,l=1}^d \tr\big(E_{ij}^*E_{kl}\big)\,\tr\big(\Phi(E_{ij})^*\Psi(E_{kl})\big) = \sum_{i,j=1}^d \tr\big(\Phi(E_{ij})^*\Psi(E_{ij})\big).
\end{equation}
The classical Marcus–Ree inequality relies on Birkhoff's theorem \cite{Bir46}; which says that permutation matrices are the only extreme points of the set of doubly stochastic matrices. This correspondence fails in the noncommutative setting, as there are extreme points of the set of unital quantum channels that are not reversible unital channels  (inner automorphisms). Consequently, an analog of classical Marcus–Ree inequality can not be expected to hold for all unital quantum channels. However, we prove that it remains valid for the class of mixed unitary channels. Note that $\langle \Phi, \Psi \rangle \geq 0$ for all CP maps on $\M{d}.$ Below is the main result of this section, which establishes the noncommutative version of the Marcus-Ree inequality.

\begin{thm}
\label{Marcus-Ree-noncom}
    Let $\Phi: \M{d} \rightarrow  \M{d}$ be a mixed unitary channel. Then, 
    \begin{equation}
    \label{Marcus-Ree-noncom-Ineq}
        \langle \Phi, \Phi \rangle \leq \max_{U \in \mcl U(d)} \langle \mathrm{Ad}_U, \Phi \rangle.
    \end{equation}
    
 We refer the inequality \eqref{Marcus-Ree-noncom-Ineq} as the \emph{noncommutative Marcus–Ree inequality}.
\end{thm}
 \begin{proof}
     Let $\Phi: \M{d} \rightarrow  \M{d}$ be a mixed unitary channel. Then, there exist $\{\lambda_j\}_{j=1}^m \subseteq [0,1]$ and $\{U_j\}_{j=1}^m \subseteq \mcl U(d)$ such that $\sum_{j=1}^m \lambda_j=1$ and $\Phi = \sum_{j=1}^m \lambda_j \mathrm{Ad}_{U_j}$. Since $\langle \Phi_1, \Phi_2 \rangle \geq 0$ for any CP maps $\Phi_1, \Phi_2$, it follows that  
     \begin{equation*}
         \langle \Phi, \Phi \rangle = \sum_{j=1}^m \lambda_j \langle \mathrm{Ad}_{U_j}, \Phi \rangle \leq \max_{U \in \mcl U(d)} \langle \mathrm{Ad}_U, \Phi \rangle \sum_{j=1}^m \lambda_j = \max_{U \in \mcl U(d)} \langle \mathrm{Ad}_U, \Phi \rangle.
     \end{equation*}
     This completes the proof.
 \end{proof}
By restricting to a special map, we recover the classical Marcus-Ree inequality.  To see this, we first need the following result.

\begin{prop}\label{prop-doubly-mixed}
    Let $A=[a_{ij}]\in\M{d}$ be a doubly stochastic matrix, and let $\Phi_{A}$ be given as in \eqref{phiA}. Then, $\Phi_A$ is a mixed unitary channel. 
\end{prop}
\begin{proof}
By Birkhoff's theorem, $A=\sum_{\ell=1}^m\lambda_{\ell}P_{\ell}$, for some permutation matrices $P_\ell$ and positive scalars $\lambda_\ell$ with $\sum_{\ell=1}^m\lambda_\ell =1$. To obtain the required conclusion, it is enough to show that 
    \begin{align*}
\Phi_A=\sum_{j=0}^{d-1}\sum_{\ell=1}^m\frac{\lambda_{\ell}}{d}\mathrm{Ad}_{V^jP_{\ell}^*},
    \end{align*}
    where $V=diag(1,\omega,\cdots, \omega^{d-1})$ and $\omega=e^{\frac{2\pi i}{d}}$.
Suppose that the permutation matrix $P_\ell$ is given by the permutation $\sigma_\ell$, i.e., $P_\ell e_j = e_{\sigma_\ell(j)}$ for all $j=1, \ldots, d$. Then, 
    $P_\ell E_{kk} P_\ell^* = E_{\sigma_\ell(k)\sigma_\ell(k)}$, and hence, for any $X=[x_{ij}]\in \M{d}$,
\begin{align*}
      \sum_{j=0}^{d-1}\sum_{\ell=1}^m\frac{\lambda_{\ell}}{d}\mathrm{Ad}_{V^jP_{\ell}^*}(X) &= \sum_{\ell=1}^m \frac{\lambda_\ell}{d} P_\ell \Big( \sum_{j=0}^{d-1} (V^*)^j X V^j  \Big)P_\ell^*\\
      &= \sum_{\ell=1}^m  \lambda_\ell P_\ell \Big( \frac{1}{d} \sum_{m,n=1}^d \sum_{j=0}^{d-1} \overline{\omega}^{(m-1)j} \omega^{(n-1)j} x_{mn}E_{mn} \Big)P_\ell^*\\
      &= \sum_{\ell=1}^m \lambda_\ell  P_\ell \Big( \sum_{k=1}^{d} x_{kk}E_{kk}  \Big)P_\ell^* 
      = \sum_{k=1}^{d}\sum_{\ell=1}^m  x_{kk}\lambda_\ell P_\ell E_{kk} P_\ell^*\\
      &= \sum_{k=1}^{d}\sum_{\ell=1}^m \lambda_\ell  x_{kk}E_{\sigma_\ell(k)\sigma_\ell(k)}
      = \sum_{i,k=1}^d \sum_{\ell=1}^m  \lambda_\ell \ip{e_i, e_{\sigma_\ell(k)}} x_{kk} E_{ii}\\
      &= \sum_{i,k=1}^d \sum_{\ell=1}^m   \ip{e_i, \lambda_\ell P_\ell e_{k}} x_{kk} E_{ii}
      = \sum_{i,k=1}^d   \ip{e_i, Ae_{k}} x_{kk} E_{ii} = \Phi_A(X).
 \end{align*} 
 This completes the proof.
\end{proof}

Below we prove the classical Marcus-Ree inequality using Theorem~\ref{Marcus-Ree-noncom}.
  \begin{cor}
      Let $A$ be a doubly stochastic matrix. Then $\ip{A,A}\leq \max_{P\in\mcl{P}_d}\ip{P,A}$.
  \end{cor}
  \begin{proof}
Suppose $A =[a_{ij}]$. Then, by \eqref{phiA}, $\Phi_A(E_{jj}) = \text{diag}(a_{1j}, \ldots, a_{dj})$ for all $j=1, \ldots, d$, and $\Phi_A(E_{ij}) = 0$ whenever $i \neq j$. Hence, 
      \begin{align}\label{eq-classical-quantum-ip-relation}
\ip{\Phi_A,\Phi_A}=\sum_{i,j=1}^d\tr\left(\Phi_A(E_{ij})^{*}\Phi_A(E_{ji})\right)=\sum_{j=1}^d\tr(\Phi_A(E_{jj})^2)=\sum_{i,j=1}^d a_{ij}^2=\ip{A,A},
      \end{align}
and for any unitary $U=[U_{ij}]$, we have
      \begin{align}\label{eq-unitary-doublystochastic}
        \ip{\mathrm{Ad}_{U},\Phi_A}=\sum_{i,j=1}^d\tr\left(\mathrm{Ad}_U(E_{ij}^*)\Phi_A(E_{ji}) \right)&=\sum_{j=1}^d\tr(U^*E_{jj}U\Phi_A(E_{jj})) =\sum_{i,j=1}^da_{ij}\tr(E_{jj}UE_{ii}U^*) \notag \\
          &=\sum_{i,j=1}^d a_{ij}\abs{\ip{e_j,Ue_i}}^2 =\ip{B(U),A},
      \end{align}
      where $B(U)$ is a doubly stochastic matrix defined by $B(U)_{ij}:=\abs{U_{ji}}^2$. By Birkhoff's theorem,  $B(U) = \sum_{k=1}^m \lambda_k P_k$, for some convex combination of permutation matrices $P_k$. Therefore,  $\ip{\mathrm{Ad}_U,\Phi_A}=\ip{B(U),A}=\ip{\sum_{k=1}^m\lambda_kP_k,A}\leq \max_{P\in\mcl{P}_d}\ip{P,A}$. Conversely, for any permutation matrix $P$, the associated doubly stochastic matrix $B(P)$ equals $P^{\T}$ and hence $\ip{P,A}=\ip{B(P^{\T}),A}=\ip{\mathrm{Ad}_{P^{\T}},\Phi_A}\leq \max_{U\in\mcl{U}(d)}\ip{\mathrm{Ad}_U,\Phi_A}$. This shows that 
      \begin{align}\label{eq-classical-quantum-ip-1}
         \max_{P\in\mcl{P}_d }\ip{P,A}=\max_{U\in\mcl{U}(d)}\ip{\mathrm{Ad}_U,\Phi_A}. 
      \end{align}
      By Proposition~\ref{prop-doubly-mixed}, $\Phi_A$ is a mixed unitary channel, and hence
      $$
      \ip{A, A} \overset{\eqref{eq-classical-quantum-ip-relation}}=\ip{\Phi_A, \Phi_A} \overset{\eqref{Marcus-Ree-noncom-Ineq}}\leq  \max_{U\in\mcl{U}(d)}\ip{\mathrm{Ad}_U,\Phi_A} \overset{\eqref{eq-classical-quantum-ip-1}}=\max_{P\in\mcl{P}_d }\ip{P,A}.
      $$
 This completes the proof.
  \end{proof}
Now, we provide an example of a unital quantum channel that fails to be mixed unitary and violates the noncommutative Marcus-Ree inequality \eqref{Marcus-Ree-noncom-Ineq}.
 
\begin{eg}\label{failure-MR-ineq}
   Let $V_1 = \frac{E_{23}- E_{32}}{\sqrt{2}}, V_2 = \frac{E_{13} - E_{31}}{\sqrt{2}}$, and $V_3 = \frac{E_{12}-E_{21}}{\sqrt{2}}.$ Consider the CP map $\Phi: \M{3} \rightarrow \M{3}$ given by
\begin{equation*}
    \Phi(X)= \sum_{j=1}^3 V_j^*XV_j =\frac{\tr(X)I_3 - X^{\T}}{2}, \quad X \in \M{3},
\end{equation*}
where $X^{\T}$ denotes the transpose of $X$.
Clearly, $\Phi$ is a unital channel but not mixed unitary \cite{Tre86}. 
Since $E_{ij}E_{kl} = \delta_{j,k}E_{il}$, $\Phi(E_{ij}) = -\frac{E_{ji}}{2} (i \neq j)$ and $\Phi(E_{ii}) = \frac{I_3 - E_{ii}}{2}$,~$i=1,2,3$, by \eqref{inner-pdt-compute},  
\begin{eqnarray*}
\ip{\Phi, \Phi} 
= \sum_{i,j=1}^3\tr\Big(\Phi(E_{ij})\Phi(E_{ji})\Big) 
&=& \frac{1}{4}\sum_{i=1}^3\tr\Big((I_3 - E_{ii})^2\Big) + \frac{1}{4}\sum_{i \neq j}\tr(E_{ji}E_{ij}) \\
&=& \frac{1}{4}\Bigg(\sum_{i=1}^3\tr(I_3 - E_{ii}) + \sum_{i \neq j} 1\Bigg)  = 3.  
\end{eqnarray*}
Also, for any $U \in \mcl U(d)$, we have
\begin{eqnarray*}
    \ip{\mathrm{Ad}_U, \Phi} 
&\overset{\eqref{inner-pdt-compute}}=& \sum_{i,j=1}^3\tr\big(U \Phi(E_{ji})U^*E_{ij}\big) 
= \frac{1}{2}\Bigg(\sum_{i=1}^3 \ip{U^*e_i, (I_3-E_{ii})U^*e_i} - \sum_{i\neq j} \ip{U^*e_j, E_{ij}U^*e_i}\Bigg)\\
&=& \frac{1}{2}\Big(3 - \sum_{i,j=1}^3 \ip{U^*e_j, E_{ij}U^*e_i}\Big) = \frac{1}{2}\Big(3 - \sum_{i,j=1}^3 \ip{e_j, Ue_i} \ip{e_j, U^*e_i}\Big)\\
&=& \frac{1}{2}(3 - \tr(\ol{U}U)) \leq \frac{1}{2} (3+1) = 2.
\end{eqnarray*}
Here, the second last inequality follows by taking $d=3$ and $(\mu_1,\mu_2,\mu_3)=(1,1,1)$ in the following fact \cite[Corollary~15]{MeWo09}:  
    given non-negative numbers $\mu_1\geq \mu_2\geq \cdots \geq \mu_d$, we have
    \begin{align*} 
        \min_{A}\{\tr(A\ol{A}): \mu(A)=\{\mu_j\}_{j=1}^d\}&= \begin{cases*}
      -2\Big(\sum_{j=1}^{\frac{d}{2}}\mu_{2j}\mu_{2j-1}\Big) & if $d$ is even,\\
     -2\Big(\sum_{j=1}^{\frac{d-1}{2}}\mu_{2j}\mu_{2j-1}\Big)+\mu_d^2 & if $d$ is odd,
    \end{cases*}
    \end{align*}
    where $\mu(A)$ denotes the set of all singular values of the matrix $A$. Thus, $$\max_{U\in\mcl{U}(3)}\ip{\mathrm{Ad}_U,\Phi}<\ip{\Phi,\Phi},$$ establishing the failure of \eqref{Marcus-Ree-noncom-Ineq}.
\end{eg}
The following is a simple observation that will be used several times later.
\begin{lem}\label{U-W-IP}
    For any $A, B \in \M{d}$, 
    \begin{equation}\label{Ad-Ip-relation}
         \ip{\mathrm{Ad}_A, \mathrm{Ad}_B} =\abs{\tr(A^*B)}^2=\abs{\ip{A,B}}^2.
    \end{equation}
\end{lem}
\begin{proof}
Let $A,B\in\M{d}$. Then, by \eqref{inner-pdt-compute},
    \begin{align*}
        \ip{\mathrm{Ad}_A,\mathrm{Ad}_B}&=\sum_{i,j=1}^d\tr\big(\mathrm{Ad}_A(E_{ij})\mathrm{Ad}_{B}(E_{ji})\big)
    =\sum_{i,j=1}^d\tr\big(A^*E_{ij}AB^*E_{ji}B)\big)\\
    &=\sum_{i,j=1}^d\tr\big(AB^*E_{ji}BA^*E_{ij}\big)
    =\sum_{i,j=1}^d\tr\big(\ranko{AB^*e_j}{AB^*e_i}E_{ij}\big)\\
    &=\sum_{i,j=1}^d\ip{AB^*e_i,e_i}\ip{e_j,AB^*e_j}
    =\abs{\tr(AB^*)}^2
    =\abs{\tr(A^*B)}^2=\abs{\ip{A,B}}^2,
    \end{align*}
 proving \eqref{Ad-Ip-relation}.
 \end{proof}

Recall that $\mathcal{MU}_d$ denotes the set of all mixed unitary channels on $\M{d}$. Define the map $\Delta_d \colon \mathcal{MU}_d \to \mathbb{R}$ by
\begin{equation*}
    \Delta_d(\Phi) := \max_{U \in \mathcal{U}(d)} \langle \mathrm{Ad}_U, \Phi \rangle - \langle \Phi, \Phi \rangle.
\end{equation*}
This map is clearly continuous. Furthermore, since the set $\mathcal{MU}_d$ is compact and convex, by Theorem~\ref{Marcus-Ree-noncom}, the range of $\Delta_d$ is a closed interval of the form $[a, b] \subset [0, \infty)$. In the next result, we determine the exact values of $ a$ and $ b$. Before we state it, recall  the completely depolarizing channel $\delta_d$ defined on $\mathbb{M}_d$  by
$$\delta_d(X) := \frac{1}{d}\tr(X)I_d, \quad X \in \M{d}.$$ Observe that for any unital or quantum channel $\Phi$, we have $\ip{\Phi,\delta_d}=1$. Indeed,
\begin{align*}
     \ip{\Phi, \delta_d} 
    \overset{\eqref{inner-pdt-compute}}= \frac{1}{d}\sum_{i,j=1}^d \tr(\Phi(E_{ij}) \tr(E_{ji}) I_d) 
    = \frac{1}{d}\sum_{i=1}^d \tr(\Phi(E_{ii})) = 1.
\end{align*}

\begin{prop}\label{prop-Delta-range}
    For each $d \geq 1$, $\Delta_d(\mathcal {MU}_d) = \left[0, \frac{d^2-1}{4}\right]$, and $\Delta_d^{-1}(\frac{d^2-1}{4})=\{\frac{\mathrm{Ad}_U+\delta_d}{2}: U\in\mathcal{U}(d)\}$.
\end{prop}

\begin{proof}
Fix $\Phi \in \mathcal{MU}_d$, and let $U_0$ be an unitary matrix such that $\max_{U \in \mathcal{U}(d)} \langle \mathrm{Ad}_U, \Phi \rangle = \langle \mathrm{Ad}_{U_0}, \Phi \rangle$. Then,
\begin{align}
\label{delta-d-phi-exp}
    \Delta_d(\Phi) &= \langle \mathrm{Ad}_{U_0}, \Phi \rangle - \langle \Phi, \Phi \rangle  
    = \langle \mathrm{Ad}_{U_0}, \Phi \rangle - \langle \Phi, \Phi \rangle + \frac{1}{4}\langle \mathrm{Ad}_{U_0}, \mathrm{Ad}_{U_0} \rangle - \frac{1}{4}\langle \mathrm{Ad}_{U_0}, \mathrm{Ad}_{U_0} \rangle \notag\\
    &= \frac{1}{4}\langle \mathrm{Ad}_{U_0}, \mathrm{Ad}_{U_0} \rangle - \left\langle \Phi - \frac{\mathrm{Ad}_{U_0}}{2}, \Phi - \frac{\mathrm{Ad}_{U_0}}{2} \right\rangle 
    \overset{\eqref{Ad-Ip-relation}}= \frac{d^2}{4} - \left\langle \Phi - \frac{\mathrm{Ad}_{U_0}}{2}, \Phi - \frac{\mathrm{Ad}_{U_0}}{2} \right\rangle.
\end{align}
Since $\langle \Psi, \delta_d \rangle = 1$ for any unital quantum channel $\Psi$, by the Cauchy-Schwarz inequality, we get
\begin{align}
\label{Cauchy-schwarz}
    \frac{1}{4} = \left| \left\langle \Phi - \frac{\mathrm{Ad}_{U_0}}{2}, \delta_d \right\rangle \right|^2 
    \leq \left\| \Phi - \frac{\mathrm{Ad}_{U_0}}{2} \right\|^2 \langle \delta_d, \delta_d \rangle
    = \left\langle \Phi - \frac{\mathrm{Ad}_{U_0}}{2}, \Phi - \frac{\mathrm{Ad}_{U_0}}{2} \right\rangle.
\end{align}
Consequently, by \eqref{delta-d-phi-exp}, we have
\begin{equation*}
    \Delta_d(\Phi) \leq  \frac{d^2 - 1}{4},
\end{equation*}
for all $\Phi \in \mathcal{MU}_d$, which implies that $\Delta_d(\mathcal{MU}_d) \subseteq \left[0, \frac{d^2-1}{4}\right]$. Furthermore, for any $U\in\mcl{U}(d)$,
\begin{align*}
    \Delta_d\left(\frac{\delta_d + \mathrm{Ad}_U}{2}\right) 
    &= \max_{W \in \mathcal{U}(d)} \left\langle \mathrm{Ad}_W, \frac{\delta_d + \mathrm{Ad}_U}{2} \right\rangle - \left\langle \frac{\delta_d + \mathrm{Ad}_U}{2}, \frac{\delta_d + \mathrm{Ad}_U}{2} \right\rangle \\
    &= \max_{W \in \mathcal{U}(d)} \left( \frac{1}{2}\langle \mathrm{Ad}_W, \delta_d \rangle + \frac{1}{2}\langle \mathrm{Ad}_W, \mathrm{Ad}_U \rangle \right) - \frac{1}{4}\big( \langle \delta_d, \delta_d \rangle + 2\langle \delta_d, \mathrm{Ad}_U \rangle + \langle \mathrm{Ad}_U, \mathrm{Ad}_U \rangle \big) \\
    & \overset{\eqref{Ad-Ip-relation}}=\Big(\frac{1}{2} + \frac{1}{2} \max_{W \in \mathcal{U}(d)}|\tr(W^*U)|^2\Big)  - \frac{1}{4}( 1 + 2 + \abs{\tr(U^*U)}^2)\\
    &=\frac{1}{2}(1 + d^2) - \frac{1}{4}(3 + d^2) = \frac{d^2 - 1}{4},
\end{align*}
and also $\Delta_d(\delta_d)=\max_{U \in \mathcal{U}(d)} \langle \mathrm{Ad}_U, \Phi \rangle - \langle \delta_d, \delta_d \rangle=1-1=0$. This shows that $0,\frac{d^2-1}{4}\in \Delta_d(\mathcal{MU}_d)$. Since $\Delta_d(\mathcal{MU}_d)$ is a compact convex set in $\left[0,\frac{d^2-1}{4}\right]$ containing $0$ and $\frac{d^2-1}{4}$, we get that  $\Delta_d(\mathcal{MU}_d)=\left[0, \frac{d^2-1}{4}\right]$. Now, suppose $\Phi\in\Delta_d^{-1}(\frac{d^2-1}{4})$. Then by \eqref{delta-d-phi-exp}, $\left\langle \Phi - \frac{\mathrm{Ad}_{U_0}}{2}, \Phi - \frac{\mathrm{Ad}_{U_0}}{2} \right\rangle=\frac{1}{4}$, which means equality holds in the Cauchy-Schwarz inequality \eqref{Cauchy-schwarz}.  Therefore, $ \Phi - \frac{\mathrm{Ad}_{U_0}}{2} =c\delta_d$, for some $c\in\mbb{C}$. Since $\Phi$ is unital, we get $c=\frac{1}{2}$ and hence $\Phi=\frac{\mathrm{Ad}_{U_0}+\delta_d}{2}$. This completes the proof.
\end{proof}

Motivated by the classical definition of Erd\H{o}s matrices (see, \cite[pg.~2]{Rag25}), we define Erd\H{o}s channels as below.
\begin{defn}\label{def-Erdos-channel}
    A mixed unitary channel $\Phi: \M{d} \rightarrow  \M{d}$ is said to be an \it{ Erd\H{o}s channel} if we have the equality in the equation \eqref{Marcus-Ree-noncom-Ineq}. That is, $\ip{\Phi, \Phi} = \max_{U \in \mcl{U}(d)} \ip{\mathrm{Ad}_U, \Phi}.$ 
    Equivalently, $\Delta_d(\Phi)=0$.
\end{defn}
\begin{rmk}\label{rmk-Erdos-Erdos channel}
    Let $A$ be a doubly stochastic matrix. By Proposition~\ref{prop-doubly-mixed}, $\Phi_{A}$ is a mixed unitary channel. Then, by \eqref{eq-classical-quantum-ip-relation} and \eqref{eq-classical-quantum-ip-1}, $\ip{\Phi_A,\Phi_A}=\ip{A,A}$ and $\max_{U\in\mcl{U}(d)}\ip{\mathrm{Ad}_U,\Phi_A}=\max_{P\in\mcl{P}(d)}\ip{P,A}$. Thus, $\Phi_A$ is an Erd\H{o}s channel if and only $A$ is an Erd\H{o}s matrix.
\end{rmk}
The well-known completely depolarizing channel $\delta_d$ is an example of an Erd\H{o}s channel. We have also seen that $\Delta_d(\frac{\mathrm{Ad}_U+\delta_d}{2})=\frac{d^2-1}{4}$, for every unitary $U$. Thus, we have examples of non-Erd\H{o}s channels, namely $\frac{\mathrm{Ad}_U+\delta_d}{2}$ for every unitary $U$.
Before presenting more examples of mixed unitary channels that are not Erd\H{o}s, we first establish a necessary condition that every Erd\H{o}s channel must satisfy.
\begin{lem}\label{erdos-kraus-attain}
Let $\Phi$ be a mixed unitary channel given by $\Phi = \sum_{j=1}^n \lambda_j \mathrm{Ad}_{U_j}$. Suppose $\Phi$ is Erd\H{o}s. Then 
\begin{align*}
    \ip{\Phi,\Phi} = \langle \Phi, \mathrm{Ad}_{U_j} \rangle, \quad j = 1, 2, \ldots, n.
\end{align*}
\end{lem}
\begin{proof}
Set $M : = \max_{U \in \mcl{U}(d)} \ip{\mathrm{Ad}_U, \Phi}.$ Then, it follows that
\begin{align*}
0\leq\sum_{j=1}^n \lambda_j (M - \langle \mathrm{Ad}_{U_j}, \Phi \rangle) = M - \sum_{j=1}^n \lambda_j \langle\mathrm{Ad}_{U_j}, \Phi \rangle = M -  \ip{\Phi, \Phi} = 0,
\end{align*}
where the last equality follows from the the fact that $\Phi$ is Erd\H{o}s. This shows that each summand is zero. That is, $\langle \mathrm{Ad}_{U_j},\Phi \rangle=M=\ip{\Phi,\Phi}$ for all $j$.
\end{proof}

Using the lemma above, we  provide a family of mixed unitary channels that are not Erd\H{o}s.
\begin{eg}
Let $d \geq 2$, and let $U \in \mathcal{U}(d)$ be such that $\tr(U) = 0$. For any $\lambda\in(0,\frac{1}{2})\cup (\frac{1}{2},1)$, consider the unital channel $\Phi_\lambda:=\lambda \id + (1-\lambda)\mathrm{Ad}_U$ on $\M{d}$.
By definition, $\Phi_\lambda$ is a mixed unitary channel. Now, we will show that $\Phi_\lambda$ is not an Erd\H{o}s channel. To see this, since $\tr(U) = 0,$ it follows that  
\begin{align*}
\ip{\Phi_\lambda, \Phi_\lambda}&= \lambda^2\ip{\id, \id} +   2\lambda (1-\lambda) \ip{\id, \mathrm{Ad}_{U}} + (1-\lambda)^2 \ip{\mathrm{Ad}_{U}, \mathrm{Ad}_{U}}\\
&\overset{\eqref{Ad-Ip-relation}}= \lambda^2 d^2 + 2\lambda(1-\lambda) |\tr(U)|^2 + (1-\lambda)^2 |\tr(U^*U)|^2 = d^2(\lambda^2+(1-\lambda)^2).
\end{align*}
On the other hand, 
\begin{align*}
\ip{\id, \Phi_\lambda}= \lambda \ip{\id, \id} +  (1-\lambda) \ip{\id, \mathrm{Ad}_{U}}
\overset{\eqref{Ad-Ip-relation}}= \lambda d^2 + (1-\lambda) |\tr(U)|^2 = \lambda d^2.
\end{align*}
Therefore, if $\Phi_\lambda$ is Erd\H{o}s, then by Lemma~\ref{erdos-kraus-attain}, $\lambda d^2 = d^2(\lambda^2+(1-\lambda)^2)$. This forces that $\lambda = \frac{1}{2}$, which contradicts the fact that $\lambda\in(0,\frac{1}{2})\cup (\frac{1}{2},1)$.  Thus, we conclude that $\Phi_\lambda$ is not Erd\H{o}s.
\end{eg}

    For any unital quantum channel $\Phi$ and $U,W\in\M{d}$, by using \cite[Lemma 1.1]{ChLi23}, $
C_{\mathrm{Ad}_U\circ\Phi\circ\mathrm{Ad}_W}=(W^{\T}\otimes U)^*C_{\Phi}(W^{\T}\otimes U)$.
    Therefore, if $U, W \in \mcl U(d)$, then
    \begin{align*}
\ip{\mathrm{Ad}_U\circ\Phi\circ\mathrm{Ad}_W,\mathrm{Ad}_U\circ\Phi\circ\mathrm{Ad}_W}=\tr\big(C_{\mathrm{Ad}_U\circ\Phi\circ\mathrm{Ad}_W}C_{\mathrm{Ad}_U\circ\Phi\circ\mathrm{Ad}_W}\big)
       =\tr\big(C_{\Phi}C_{\Phi}\big)=\ip{\Phi,\Phi},
    \end{align*}
    and
    \begin{align*}
\max_{V\in\mcl{U}(d)}\ip{\mathrm{Ad}_V,\mathrm{Ad}_U\circ\Phi\circ\mathrm{Ad}_W}=\max_{V\in\mcl{U}(d)}\ip{\mathrm{Ad}_{U^*}\circ\mathrm{Ad}_V\circ\mathrm{Ad}_{W^*},\Phi}
        &=\max_{V\in\mcl{U}(d)}\ip{\mathrm{Ad}_{W^*VU^*},\Phi}\\
        &=\max_{Z\in\mcl{U}(d)}\ip{\mathrm{Ad}_Z,\Phi}.
    \end{align*}
This shows that 
\begin{equation} \label{Erdos-preserve-equivalence}
\text{$\Phi$ is Erd\H{o}s if and only if $\mathrm{Ad}_U \circ \Phi \circ \mathrm{Ad}_W$ is Erd\H{o}s for any unitaries $U,W$.}  
\end{equation}
This observation leads to a natural equivalence relation on unital channels as follows. 
\begin{defn}\label{equivalence-reln}
    Let $\Phi, \Psi$ be two unital quantum channels on $\M{d}$. We say that $\Phi$ and $\Psi$ are unitarily equivalent if there exists $U, W \in \mcl U(d)$ such that $\Phi = \mathrm{Ad}_U \circ \Psi \circ \mathrm{Ad}_W$.
\end{defn}
It is worth noting from Proposition~\ref{prop-Delta-range} that mixed unitary channels $\Phi$ satisfying $\Delta_d(\Phi) = \frac{d^2-1}{4}$ admit only a single equivalence class. We next investigate mixed unitary channels $\Phi$ satisfying $\Delta_d(\Phi) = 0$. Equivalently, classify all the Erd\H{o}s channels up to unitary equivalence on $\M{d}$.
        \section{Erd\H{o}s property of conditional expectations}\label{S4}

In this section, we investigate the Erd\H{o}s property of trace-preserving conditional expectation associated with a unital $C^*$-subalgebra of $\mathbb{M}_d$. We begin by recalling the formal definition of a conditional expectation.

\begin{defn}\label{conditional-exp}
    Let $\mcl{A}$ be a unital $C^*$-subalgebra of $\M{d}$. A linear map $\Psi$ on $\M{d}$ is said to be a \emph{conditional expectation} onto $\mcl{A}$ if \begin{enumerate}[label=(\roman*)]
        \item $\Psi$ is unital, CP map;
        \item $\Psi(A)=A \quad \text{for all}~A\in\mcl{A}$;
        \item $\Psi(\M{d})=\mcl{A}$.
    \end{enumerate}
\end{defn}
It is well known that any conditional expectation map $\Psi$ is an idempotent ($\Psi^2=\Psi$), self-dual ($\Psi^*=\Psi$) map which has a bi-module property:
    $\Psi(XYZ)=X\Psi(Y)Z$ for $X,Z\in\mcl{A}, Y\in\M{d}$.
Note that a conditional expectation need not be unique. However, if $\Psi$ is also trace-preserving, then $\Psi$ must be unique and equal to the orthogonal projection of $\M{d}$ onto $\mcl{A}$ with respect to the Hilbert-Schmidt inner product, denoted as $\mbf{E}_{\mcl{A}}$. It is known  \cite{KJRM23} that $\mbf{E}_{\mcl{A}}$ is a mixed unitary channel.
In particular, for the unital $ C^*$-subalgebra $\mcl{A}=\mbb{C}I_d$, the unique trace-preserving conditional expectation onto $\mbb{C}I_d$ is given by the completely depolarizing channel  $\delta_d$. Since $\delta_d$ is an Erd\H{o}s channel, it is then natural to ask whether, for a given unital $C^*$-subalgebra $\mcl{A}$ of $\M{d}$, is it true that the unique trace-preserving conditional expectation map $\mbf{E}_{\mcl{A}}$  Erd\H{o}s?  
We establish an affirmative answer to this question in Theorem~\ref{cond-exp-erdos}. Before we state and prove it, we need some preparatory results.

Let $\mcl{A}$ be a unital $C^*$-subalgebra of $\M{d}$. Without loss of generality, assume that $$\mcl{A}=\oplus_{k=1}^\ell (I_{m_k}\otimes \M{n_k}),$$ where 
   $I_{m_k}\otimes \M{n_k}:=\{I_{m_k}\otimes X\in\M{m_kn_k}: X\in\M{n_k}\},$ and $n_k,m_k\in\mbb{N}$ such that $\sum_{k=1}^\ell n_k m_k=d$. Let $\mcl{A}'$ denote the commutant of $\mcl{A}$ in $\M{d}$. Then, $\mcl{A}' = \oplus_{k=1}^\ell (\M{m_k}\otimes I_{n_k})$. Let $\mathcal{U}(\mathcal{A}')$ denote the group of unitaries in $\mathcal{A}'$. Then, it is compact and equal to $\oplus_{k=1}^\ell (\mcl{U}(m_k) \otimes I_{n_k})$. Furthermore, $\mathcal{U}(\mathcal{A}')$ is naturally isomorphic to the compact group $\prod_{k=1}^\ell \mcl{U}(m_k)$. Let $\mu$ be the normalized Haar measure on $\mcl{U}(\mcl{A}')$. Under the above identification, $\mu$ is equal to the product measure $\mu_1\otimes\mu_2\otimes\cdots\mu_\ell$, where $\mu_k$ is the normalized Haar measure on $\mcl{U}(m_k),1\leq k\leq \ell$. Consequently, for every continuous function $f:\mcl{U}(\mcl{A}')\to \mbb{C}$ we have
   \begin{align*}
       \int\limits_{U\in\mcl{U}(\mcl{A}')}f(U) d\mu(U)=\int\limits_{U_1\in\mcl{U}(m_1)}\int\limits_{U_2\in\mcl{U}(m_2)}\cdots\int\limits_{U_\ell\in\mcl{U}(m_\ell)}f(U)d\mu_1(U_1)d\mu_2(U_2)\cdots d\mu_\ell(U_\ell)
   \end{align*}
 In particular, if $f(U)=f_1(U_1)f_2(U_2)\cdots f_\ell(U_\ell)$, where $f_i:\mcl{U}(m_i)\to\mbb{C}$ is continuous. Then,
 \begin{align*}
    \int\limits_{U\in\mcl{U}(\mcl{A}^{'})}f(U) d\mu(U)= \prod_{k=1}^\ell\int\limits_{U_k\in\mcl{U}(m_k)} f_k(U_k)d\mu_k(U_k).
 \end{align*}

    Given a unital quantum channel $\Phi$,  we associate a linear space $S_{\Phi}$, called the \emph{Kraus space} of $\Phi$;  defined as $S_{\Phi}:=\lspan\{K_j: 1\leq j\leq m\}$, where $\{K_1,K_2,\cdots, K_m\}$ is  any set of Kraus operators of $\Phi$. This definition is well defined. Indeed, if $\{K_1,K_2,\cdots, K_m\}$ and $\{L_1,L_2,\cdots,L_n\}$ be two set of Kraus operators for $\Phi$ then $\lspan\{K_1,K_2,\cdots,K_m\}=\lspan\{L_1,L_2,\cdots,L_n\}$. Therefore, $S_{\Phi}$ is independent of a given Kraus representation of $\Phi$.
    
\begin{lem}\label{lem-mix-unit-int-rep}
    Let $\Phi$ be a mixed unitary channel and let $S_{\Phi}$ denote its Kraus space. Then there exists a probability measure $\mu$ on $S_{\Phi}\cap\mcl{U}(d)$ such that 
    \begin{align*}
        \Phi(X)=\int_{U\in S_{\Phi}\cap\mcl{U}(d)}U^*XU d\mu(U), \quad  X\in\M{d}.
    \end{align*}
\end{lem}
\begin{proof}
    Let $\Phi$ be a mixed unitary channel given by $\Phi = \sum_{j=1}^m \lambda_j \mathrm{Ad}_{U_j}$. Define a  measure $\mu := \sum_{j=1}^m \lambda_j \delta_j$, where $\delta_j$ is the point mass measure supported at $U_j$. Then, $\mu(\mathcal U(d)) = \sum_{j=1}^m \lambda_j=1$, so that $\mu$ is a probability measure on $\mcl{U}(d)$. Since each $U_j \in S_\Phi \cap \mathcal U(d)$ for all $j =1, \ldots, m$, we have that $\mu$ is a probability measure on $S_{\Phi}\cap\mcl{U}(d)$. Furthermore, 
  $$
  \int_{S_\Phi \cap \mathcal U(d)} U^*XU d\mu(U) = \sum_{j=1}^m \lambda_j \int_{S_\Phi \cap \mathcal U(d)} U^*XU d\delta_j(U) = \sum_{j=1}^m \lambda_j U_j^*XU_j=\Phi(X).
  $$
  This completes the proof.
\end{proof}

 \begin{rmk}\cite{KJRM23}
     Let $\mu$ be the normalized (left invariant) Haar measure on the unitary group $\mcl{U}(d)$. Then
     \begin{align}\label{eq-deltad-int-form}
         \int_{U \in \mcl{U}(d)}\mathrm{Ad}_{U} d\mu(U)=\delta_d.
     \end{align}
     This also follows from Lemma~\ref{lem-mix-unit-int-rep} along with the fact that $S_{\delta_d} = \M{d}$. 
 \end{rmk}
 
\begin{lem}\label{eq-Haar-int-zero}
   Let $\mu$ be the normalized (left invariant) Haar measure on the unitary group $\mcl{U}(d)$ and let $A\in\M{d}$. Then,
 $$
  \int\limits_{U\in\mcl{U}(d)} \tr(AU)d\mu(U) =0 \quad \text{and} \quad   \int\limits_{U\in\mcl{U}(d)} \abs{\tr(AU)}^2d\mu(U)=\frac{1}{d}\tr(A^*A).
  $$
\end{lem}
 \begin{proof}
     By invariance of Haar measure, it follows for any $W\in\mcl{U}(d)$ that 
     $$\int\limits_{U\in\mcl{U}(d)} \tr(AU)d\mu(U)=\int\limits_{U\in\mcl{U}(d)} \tr(AWU)d\mu(U).$$ In particular, by taking $W=-I_d$, we get $\int\limits_{U\in\mcl{U}(d)} \tr(AU)d\mu(U)=-\int\limits_{U\in\mcl{U}(d)} \tr(AU)d\mu(U)$, which proves the first identity. To see the remaining part, note that
     \begin{align*}
         \ip{\mathrm{Ad}_{A^*},\delta_d}=\sum_{i,j=1}^d\tr\big(\mathrm{Ad}_{A^*}(E_{ij})\delta_d(E_{ji})\big)=\frac{1}{d}\sum_{i=1}^d\tr(AE_{ii}A^*)
         =\frac{1}{d}\sum_{i=1}^d\ip{Ae_i,Ae_i}=\frac{1}{d}\tr(A^*A).
     \end{align*} 
Therefore,
\begin{align*}
       \int\limits_{U\in\mcl{U}(d)} \abs{\tr(AU)}^2d\mu(U) \overset{\eqref{Ad-Ip-relation}} =  \int\limits_{U\in\mcl{U}(d)} \ip{\mathrm{Ad}_{A^*},\mathrm{Ad}_U}d\mu(U) &= \ip{\mathrm{Ad}_{A^*},\int\limits_{U\in\mcl{U}(d)} \mathrm{Ad}_U d\mu(U)}\\ &\overset{\eqref{eq-deltad-int-form}}=  \ip{\mathrm{Ad}_{A^*},\delta_d} =  \frac{1}{d}\tr(A^*A).
     \end{align*}
This completes the proof.
 \end{proof}

 \begin{lem}
    Let $\Phi:\M{d_1}\to\M{d_2}$ be a unital completely positive  map. Then $\norm{\Phi(X)}_2^2\leq \tr(\Phi(X^*X))$ for all $X \in \M{d_1}$. In particular, 
    \begin{align}\label{eq-partial-trace-contr}
        \norm{(\id_m \otimes \tr_n)(Y)}_2^2\leq n\norm{Y}_2^2, \quad Y\in\M{m}\otimes \M{n}.
    \end{align}
\end{lem}

\begin{proof}
    By the Kadison–Schwarz inequality, we have $\Phi(X)^*\Phi(X)\leq \Phi(X^*X)$ for all $X \in \M{d_1}$. Taking the trace on both sides yields
    \begin{align*}
        \norm{\Phi(X)}_2^2=\tr(\Phi(X)^*\Phi(X))\leq \tr(\Phi(X^*X)).
    \end{align*}
    For the remaining part, consider the unital completely positive map $\Phi: \M{mn} \to \M{m}$ given by $\Phi = \frac{1}{n}(\id_m \otimes \tr_n)$. Then
    \begin{align*}
        \norm{\frac{1}{n}(\id_m \otimes \tr_n)(Y)}_2^2 \leq \tr\left(\frac{1}{n}(\id_m \otimes \tr_n)(Y^*Y)\right) = \frac{1}{n}\tr(Y^*Y) = \frac{1}{n}\norm{Y}_2^2, \quad Y \in \M{mn}.
    \end{align*}
    This completes the proof.
\end{proof}
\begin{thm}
\label{cond-exp-erdos}
    Let $\mcl{A}\subseteq \M{d}$ be a unital $C^*$-algebra and let $\mbf{E}_{\mcl{A}}$ be the trace-preserving conditional expectation onto $\mcl{A}$. Then $\mbf{E}_{\mcl{A}}$ is an Erd\H{o}s channel.
\end{thm}
\begin{proof}
     Without loss of generality, assume that $\mcl{A}=\oplus_{k=1}^\ell (I_{m_k}\otimes \M{n_k})$, where $n_k,m_k\in\mbb{N}$ such that $\sum_{k=1}^\ell n_k m_k=d$. Then $\mbb{C}^d=\bigoplus_{j=1}^\ell (\mbb{C}^{n_j}\otimes \mbb{C}^{m_j})$, and any $W\in\M{d}$ is given by $W=[W_{ij}]_{i,j=1}^\ell \in\M{d}$, where $W_{ij}\in\M{m_in_i,m_jn_j}=\B{\mbb{C}^{n_j}\otimes \mbb{C}^{m_j},\mbb{C}^{n_i}\otimes \mbb{C}^{m_i}}$. Since any $U\in\mcl{U}(\mcl{A}^{'})$ is of the form $\oplus_{k=1}^\ell(U_k\otimes I_{n_k})$ for some $U_k\in\mcl{U}(m_k),~k=1, \ldots, \ell$, we have
  \begin{align*}
      \tr(WU)=\tr\Big(\sum_{k=1}^\ell W_{kk}U_{k}\Big)
             =\sum_{k=1}^\ell \tr(W_kU_k),
  \end{align*}
  where $W_k=(\id_{m_k}\otimes \tr_{n_k})(W_{kk})\in\M{m_k}$, and
\begin{align}\label{eq-tr-WU}
  \abs{\tr(WU)}^2=\sum_{j,k=1}^\ell \tr(W_jU_j)\tr(W_k^*U_k^*)=\sum_{k=1}^\ell \abs{\tr(W_kU_k)}^2+ \sum_{\substack{j, i=1 \\j\neq i}}^\ell \tr(W_jU_j)\tr(W_i^*U_i^*).
\end{align}
 Now, from \cite[Lemma 12]{KJRM23}, the map $\mbf{E}_{\mcl{A}}$ is a mixed unitary channel, and is given by
    \begin{align*}
        \mbf{E}_{\mcl{A}}(X)=\int\limits_{U\in\mcl{U}(\mcl{A}^{'})} U^*XU d\mu(U), \quad X\in\M{d}.
    \end{align*}
Therefore, for any unitary $W=[W_{ij}]_{i,j=1}^\ell \in\M{d}$, by Lemma~\ref{eq-Haar-int-zero},
    \begin{eqnarray*}
        \ip{\mathrm{Ad}_{W^*},\mbf{E}_{\mcl{A}}}&=&\int\limits_{U\in\mcl{U}(\mcl{A}^{'})}\ip{\mathrm{Ad}_{W^*},\mathrm{Ad}_U} d\mu(U)
         = \int\limits_{U\in\mcl{U}(\mcl{A}^{'})} \abs{\tr(WU)}^2 d\mu(U)\\
        &\overset{\eqref{eq-tr-WU}}=&\sum_{k=1}^\ell \int\limits_{U\in\mcl{U}(\mcl{A}^{'})}\abs{\tr(W_kU_k)}^2 d\mu(U)+ \sum_{\substack{j, i=1 \\j\neq i}}^\ell \int\limits_{U\in\mcl{U}(\mcl{A}^{'})}\tr(W_jU_j)\tr(W_i^*U_i^*) d\mu(U)\\
        &=&\sum_{k=1}^\ell \int\limits_{U\in\mcl{U}(m_k)}\abs{\tr(W_kU_k)}^2 d\mu(U_k)+\sum_{\substack{j, i=1 \\j\neq i}}^\ell \int\limits_{U\in\mcl{U}(\mcl{A}^{'})}\tr(W_jU_j)\tr(W_i^*U_i^*) d\mu(U)\\
       &=& \sum_{k=1}^{\ell}\frac{\tr(W_k^*W_k)}{m_k} +\sum_{\substack{j, i=1 \\j\neq i}}^\ell \Big(\int\limits_{U\in\mcl{U}(m_j)} \tr(W_jU_j) d\mu(U_j)\Big) \Big(\int\limits_{U\in\mcl{U}(m_i)} \tr(W_i^*U_i^*) d\mu(U_i)\Big)\\
        &=& \sum_{k=1}^{\ell}\frac{1}{m_k} \tr(W_k^*W_k).
    \end{eqnarray*}
 In particular, if $W=I_d$, then $\sum_{k=1}^{\ell}\frac{1}{m_k} \tr(W_k^*W_k)=\sum_{k=1}^{\ell}\frac{1}{m_k} \tr(n_k^2I_{m_k})=\sum_{k=1}^{\ell}n_k^2$. This together with the fact $\mbf{E}_{\mcl{A}}^*=\mbf{E}_{\mcl{A}}=\mbf{E}_{\mcl{A}}^2$ gives that
    \begin{eqnarray}\label{eq-ip of E A-dim A}
 \ip{\mbf{E}_{\mcl{A}},\mbf{E}_{\mcl{A}}}=\ip{\id,\mbf{E}_{\mcl{A}}}=\sum_{k=1}^{\ell}n_k^2=\dim(\mcl{A}).
    \end{eqnarray}
   On the other hand, if $W$ is any unitary, that is  $W^*W=I_d$. In this identity, by comparing $(k,k)^{th}$ entry, we get that
$\sum_{j=1}^\ell W_{jk}^*W_{jk}=I_{m_k}\otimes I_{n_k}$. Therefore, 
\begin{align*}
  \tr(W_{kk}^*W_{kk})\leq \tr\Big(\sum_{j=1}^\ell W_{jk}^*W_{jk}\Big)=\tr(I_{m_k}\otimes I_{n_k})=n_km_k.
\end{align*}
Consequently,
\begin{align*}
    \ip{\mathrm{Ad}_{W^*},\mbf{E}_{\mcl{A}}}=\sum_{k=1}^{\ell}\frac{1}{m_k} \tr(W_k^*W_k)
    \overset{\eqref{eq-partial-trace-contr}}\leq \sum_{k=1}^{\ell}\frac{n_k}{m_k} \tr(W_{kk}^*W_{kk})\leq \sum_{k=1}^{\ell}\frac{n_k}{m_k} n_km_k=\sum_{k=1}^{\ell}n_k^2.
\end{align*}
Thus, $\max_{W \in \mcl{U}(d)}\ip{\mathrm{Ad}_{W},\mbf{E}_{\mcl{A}}}=\sum_{k=1}^{\ell}n_k^2=\ip{\mbf{E}_{\mcl{A}},\mbf{E}_{\mcl{A}}}$, completing the proof.
\end{proof}

\begin{eg}(Pinching channel)
Let $\Phi$ be a unital channel on $\M{d}$ given by
\begin{equation*}
    \Phi(X) = \sum_{j=1}^n P_j X P_j,
\end{equation*}
where $\{P_j\}_{j=1}^n$ is a collection of mutually orthogonal projections satisfying $\sum_{j=1}^n P_j = I_d$. It is known that $\Phi$ is a mixed unitary channel (see \cite{Wat18}). We claim that $\Phi$ is an Erd\H{o}s channel. To see this, consider the algebra 
\begin{equation*}
    \mcl{A} := \{X \in \M{d} : X P_j = P_j X \text{ for all } j=1, \ldots, n\}.
\end{equation*}
Clearly, $\Phi$ fixes elements of $\mcl{A}$ and the range of $\Phi$ is equal to $\mcl{A}$. Thus, $\Phi$ is a trace-preserving conditional expectation onto $\mcl{A}$ and hence, by the uniqueness of the trace-preserving conditional expectation, $\mbf{E}_{\mcl{A}} = \Phi$. Now, by Theorem~\ref{cond-exp-erdos}, the map $\Phi$ is an Erd\H{o}s channel. 
\end{eg}

Our next objective in this section is to show that the unique trace-preserving conditional expectation $\mbf{E}_{\mcl{A}}$ is locally isolated within its Choi-rank class. Before stating and proving this result, we first establish a couple of lemmas and propositions.

Given a unital $C^*$-subalgebra $\mcl{A}$ of $\M{d}$, we let $F(\mcl{A})$ denote the set of all unital quantum channels that fixes elements of $\mcl{A}$. When $\Phi\in F(\mcl{A})$, the noncommutative Marcus–Ree  inequality \eqref{Marcus-Ree-noncom-Ineq} can be sharpened as follows.

\begin{prop} \label{Refined-marcus-ree}
Let $\mcl{A}$ be a unital $C^*$-subalgebra of $\M{d}$. If $\Phi \in F(\mcl{A})$ is a mixed unitary channel, then 
\begin{align*}
\ip{\Phi, \Phi} \leq \max_{U \in \mcl{U}(\mcl{A}')} \ip{\mathrm{Ad}_U, \Phi}.
\end{align*}
Here $\mcl{U}(\mcl{A}'):=\mcl{U}(d)\cap \mcl{A}'$. Furthermore, if $\Phi$ is an Erd\H{o}s channel, then $\ip{\Phi, \Phi} = \max_{U \in \mcl{U}(\mcl{A}')} \ip{\mathrm{Ad}_U, \Phi}.$
\end{prop}
\begin{proof}
Let $\Phi$ be a mixed unitary channel given by $\Phi = \sum_{j=1}^\ell \lambda_j \mathrm{Ad}_{U_j}$, where $\lambda_j > 0$, $\sum_{j=1}^\ell \lambda_j = 1$, and $U_j \in \mcl{U}(d)$. By setting $V_j := \sqrt{\lambda_j} U_j$ we write $\Phi=\sum_{j=1}^\ell \mathrm{Ad}_{V_j}$ with $\sum_{j=1}^\ell V_j^* V_j = I$. Since $\Phi \in F(\mcl{A})$, for any $A \in \mcl{A}$, we have
\begin{align*}
\sum_{j=1}^\ell (V_j A - A V_j)^* (V_j A - A V_j) &= \sum_{j=1}^\ell \left( A^* V_j^* V_j A - A^* V_j^* A V_j - V_j^* A^* V_j A + V_j^* A^* A V_j \right) \\
&= A^* \big( \sum_{j=1}^\ell V_j^* V_j \big) A - A^* \Phi(A) - \Phi(A^*) A + \Phi(A^* A) =0.
\end{align*}
Because each individual term $(V_j A - A V_j)^*(V_j A - A V_j)$ is positive semi-definite, their sum is zero if and only if each term is identically zero. Thus, $V_j A - A V_j = 0$ for each $j$ and for all $A\in\mcl{A}$. This implies that each $V_j$ and, consequently, each unitary $U_j$ must be in $\mcl{A}'$.  Thus, we obtain
\begin{align*}
\ip{\Phi, \Phi} = \sum_{j=1}^\ell \lambda_j \ip{\mathrm{Ad}_{U_j}, \Phi} \leq \sum_{j=1}^\ell \lambda_j \max_{U \in \mcl{U}(\mcl{A}')} \ip{\mathrm{Ad}_U, \Phi} = \max_{U \in \mcl{U}(\mcl{A}')} \ip{\mathrm{Ad}_U, \Phi}.
\end{align*}
Furthermore, if $\Phi$ is an Erd\H{o}s channel, then  
\begin{align*}
    \max_{U \in \mcl{U}(d)} \ip{\mathrm{Ad}_U, \Phi}&=\ip{\Phi,\Phi}
    \leq  \max_{U \in \mcl{U}(\mcl{A}')} \ip{\mathrm{Ad}_U, \Phi}
    \leq \max_{U \in \mcl{U}(d)} \ip{\mathrm{Ad}_U, \Phi}.
\end{align*}
This concludes the proof.
\end{proof}

\begin{lem}\label{lem-ip-grea}
    Let $\mcl{A}$ be a unital $C^*$-subalgebra of $\M{d}$. Suppose $\Phi\in F(\mcl{A})$. Then, 
    \begin{eqnarray}
    \label{exp-ip-channel}
        \ip{\Phi, \mbf{E}_{\mcl{A}}} = \dim(\mcl{A}).
    \end{eqnarray}
    Consequently, $\dim(\mcl{A}) \leq \ip{\Phi,\Phi}$, and if $\Phi\neq \mbf{E}_{\mcl{A}}$, then $\dim(\mcl{A})< \ip{\Phi,\Phi}$.
\end{lem}
\begin{proof}
      Since $\Phi$ fixes the elements of $\mcl{A}$, the composition $\mbf{E}_{\mcl{A}} \circ \Phi$ is a conditional expectation onto $\mcl{A}$. Furthermore, because $\mbf{E}_{\mcl{A}} \circ \Phi$ is trace-preserving, the uniqueness of the trace-preserving conditional expectation implies that $\mbf{E}_{\mcl{A}} \circ \Phi = \mbf{E}_{\mcl{A}}$. Using the self-adjointness property $\mbf{E}_{\mcl{A}} = \mbf{E}_{\mcl{A}}^*$, we obtain
\begin{equation*} 
\ip{\Phi, \mbf{E}_{\mcl{A}}} = \ip{\mbf{E}_{\mcl{A}} \circ \Phi, \id} = \ip{\mbf{E}_{\mcl{A}}, \id}\overset{\eqref{eq-ip of E A-dim A}} = \dim(\mcl{A}),
\end{equation*}
which establishes \eqref{exp-ip-channel}. Combining the identity \eqref{exp-ip-channel} with the Cauchy–Schwarz inequality yields
\begin{equation} \label{dim-cauchy-ineq}
\dim(\mcl{A})^2 = \abs{\ip{\Phi, \mbf{E}_{\mcl{A}}}}^2 \leq \ip{\Phi, \Phi} \ip{\mbf{E}_{\mcl{A}}, \mbf{E}_{\mcl{A}}} = \ip{\Phi, \Phi} \dim(\mcl{A}),
\end{equation}
which simplifies to $\dim(\mcl{A}) \leq \ip{\Phi, \Phi}$. Finally, equality holds in \eqref{dim-cauchy-ineq} if and only if $\Phi$ is a scalar multiple of $\mbf{E}_{\mcl{A}}$. Since both $\Phi$ and $\mbf{E}_{\mcl{A}}$ are unital, it follows immediately that $\Phi = \mbf{E}_{\mcl{A}}$.
\end{proof}

\begin{prop}
     Let $\mcl{A}$ be a unital $C^*$ subalgebra of $\M{d}$ and let $\mbf{E}_{\mcl{A}}$ be the unique trace-preserving conditional expectation onto $\mcl{A}$. Then, the Choi-rank of $\mbf{E}_{\mcl{A}}$ is $\dim(\mcl{A}')$.
\end{prop}
\begin{proof}
Without loss of generality, assume that $\mcl{A}=\bigoplus_{k=1}^\ell (I_{m_k}\otimes \M{n_k})$, where the natural numbers $m_k, n_k$ satisfy $\sum_{k=1}^\ell m_kn_k=d$. Then $\mcl{A}'=\bigoplus_{k=1}^\ell (\M{m_k}\otimes I_{n_k})$. For computational simplicity, we restrict our proof to the case $\ell=2$, the general case follows by a similar argument. 

Let $\mathbb C^d=\mcl{H}_1\oplus \mcl{H}_2$, where $\mcl{H}_k=\mbb{C}^{m_k}\otimes \mbb{C}^{n_k}$ for $k=1,2$, and let $$\bigcup_{k=1}^2\{e_r^{(k)}\otimes f_s^{(k)}: 1\leq r\leq m_k, 1\leq s\leq n_k\}$$ be the standard orthonormal basis of $\mbb C^d$. The associated matrix units of $\B{\mbb C^d} = \M{d}$ are then given by $$\bigcup_{j,k=1}^2\Big\{E_{rr'}^{(j,k)}\otimes F_{ss'}^{(j,k)}: 1\leq r\leq m_j,1\leq r'\leq m_k,1\leq s\leq n_j, 1\leq s'\leq n_k\Big\},$$ where $E_{rr'}^{(j,k)}=\ranko{e_r^{(j)}}{e_{r'}^{(k)}}$ and $F_{ss'}^{(j,k)}=\ranko{f_s^{(j)}}{f_{s'}^{(k)}}$. 

Any operator $A \in \B{\mbb C^d}$ can be expressed in the block form 
$A=\Matrix{A_{11}&A_{12}\\A_{21}&A_{22}}$, where $A_{jk}\in\B{\mcl{H}_k,\mcl{H}_j}$ for $j,k=1,2$. Under this decomposition, the map $\mbf{E}_{\mcl{A}}$ is given by 
\begin{align*}    \mbf{E}_{\mcl{A}}\Bigg(\Matrix{A_{11}&A_{12}\\A_{21}&A_{22}}\Bigg)=\Matrix{\frac{I_{m_1}}{m_1}\otimes (\tr\otimes \id)(A_{11})&0\\0&\frac{I_{m_2}}{m_2}\otimes (\tr\otimes \id)(A_{22})}.
\end{align*}
Then the Choi-matrix of $\mbf{E}_{\mcl{A}}$ is given by
\begin{align*}
    C_{\mbf{E}_{\mcl{A}}} &=  \sum_{j,k=1}^2\sum_{r,r'=1}^{m_j,m_k}\sum_{s,s'=1}^{n_j,n_k}E_{rr'}^{(j,k)}\otimes F_{ss'}^{(j,k)}\otimes \mbf{E}_{\mcl{A}}\big(E_{rr'}^{(j,k)}\otimes F_{ss'}^{(j,k)}\big).
\end{align*}
Because the conditional expectation $\mbf{E}_{\mcl{A}}$ completely annihilates the off-diagonal blocks, we have $\mbf{E}_{\mcl{A}}\big(E_{rr'}^{(j,k)}\otimes F_{ss'}^{(j,k)}\big) = 0$ whenever $j \neq k$. For the diagonal blocks, a direct application of the map yields
\begin{align*}
    \mbf{E}_{\mcl{A}}\big(E_{rr'}^{(1,1)}\otimes F_{ss'}^{(1,1)}\big) = \delta_{r,r'} \frac{I_{m_1}}{m_1} \otimes f_{ss'}^{(1,1)} ~~\text{and}~~
    \mbf{E}_{\mcl{A}}\big(E_{rr'}^{(2,2)}\otimes F_{ss'}^{(2,2)}\big) = \delta_{r,r'} \frac{I_{m_2}}{m_2} \otimes f_{ss'}^{(2,2)},
\end{align*}
where $\delta_{r,r'}$ denotes the Kronecker delta and $f_{ss'}^{(j,j)}$ are the standard matrix units of $\M{n_j}$. Substituting these expressions into the formula for $C_{\mbf{E}_{\mcl{A}}}$, the matrix evaluates to
\begin{align*}
    C_{\mbf{E}_{\mcl{A}}} = \Matrix{\sum\limits_{r=1}^{m_1}\sum\limits_{s,s'=1}^{n_1}e_{rr}^{(1,1)}\otimes f_{ss'}^{(1,1)}\otimes \frac{I_{m_1}}{m_1}\otimes f_{ss'}^{(1,1)} & 0 \\ 0 & \sum\limits_{r=1}^{m_2}\sum\limits_{s,s'=1}^{n_2}e_{rr}^{(2,2)}\otimes f_{ss'}^{(2,2)}\otimes \frac{I_{m_2}}{m_2}\otimes f_{ss'}^{(2,2)}}.
\end{align*}
By swapping the second and third tensor factors within each non-zero block entry, we obtain
\begin{align*}
    \Matrix{\Big(\sum\limits_{r=1}^{m_1} e_{rr}^{(1,1)} \otimes \frac{I_{m_1}}{m_1}\Big) \otimes \Big(\sum\limits_{s,s'=1}^{n_1} f_{ss'}^{(1,1)} \otimes f_{ss'}^{(1,1)}\Big) & 0 \\ 0 & \Big(\sum\limits_{r=1}^{m_2} e_{rr}^{(2,2)} \otimes \frac{I_{m_2}}{m_2}\Big) \otimes \Big(\sum\limits_{s,s'=1}^{n_2} f_{ss'}^{(2,2)} \otimes f_{ss'}^{(2,2)}\Big)}.
\end{align*}
Using the identity $\sum_{r} e_{rr}^{(j,j)} = I_{m_j}$ and defining $f^{(j)}:=\sum_{s}f_s^{(j)}\otimes f_s^{(j)}$ for $j=1,2$, this block matrix simplifies to
\begin{align*}
    \Matrix{\frac{I_{m_1}\otimes I_{m_1}}{m_1} \otimes \ranko{f^{(1)}}{f^{(1)}} & 0 \\ 0 & \frac{I_{m_2}\otimes I_{m_2}}{m_2} \otimes \ranko{f^{(2)}}{f^{(2)}}}.
\end{align*}
Thus, the Choi-matrix of $\mbf{E}_{\mcl{A}}$ is permutation equivalent to the matrix above. Consequently, the Choi rank of $\mbf{E}_{\mcl{A}}$ is equal to the rank of this matrix, which is precisely  $\sum_{j=1}^2 m_j^2 = \dim(\mcl{A}')$.
\end{proof}

\begin{lem}\cite{BhDe26} \label{lem-mix uni-ball}
  Let $\mcl{A}$ be a unital $C^*$-subalgebra of $\M{d}$ and let $\mbf{E}_{\mcl{A}}$ be the unique trace-preserving conditional expectation onto $\mcl{A}$. Then, there exists $r_{\mcl{A}} >0$ such that every $\Phi \in F(\mcl{A})$ satisfying $\|\Phi -\mbf{E}_{\mcl{A}}\| \leq r_{\mcl{A}}$ is a mixed unitary channel.
\end{lem}

\begin{thm}
\label{nonErdos-suff-cond}
   Let $\mcl{A}$ be a unital $C^*$-subalgebra of $\M{d}$. Suppose $\Phi\in F(\mcl{A})$ is a mixed unitary channel with Choi-rank equal to $\dim(\mcl{A}')$ and  $0<\norm{\Phi-\mbf{E}_{\mcl{A}}}\leq r_{\mcl{A}}$, where $r_{\mcl{A}}$ is as in Lemma \ref{lem-mix uni-ball}. Then $\Phi$ is not an Erd\H{o}s channel.
\end{thm}  
\begin{proof}
    Fix $\epsilon \in (0,1)$ (the suitable choice will be made later), consider the map $\eta = \frac{1}{1-\epsilon}(\Phi - \epsilon \mbf{E}_{\mcl{A}})$ on $\M{d}$. Without loss of generality, assume that $\mcl{A}=\oplus_{k=1}^\ell (I_{m_k}\otimes \M{n_k})$. Since $\Phi,\mbf{E}_{\mcl{A}}\in F(\mcl{A})$, the map $\eta$ fixes the elements of $\mcl{A}$. Now, from the proof of \cite[Theorem 3.4]{BhDe26}, there exists $V_{jk}\in\M{m_k}, (1\leq j\leq\dim(\mcl{A}'), 1\leq k\leq \ell)$ such that  
    \begin{align*}
        \Phi([X_{ij}]_{i,j=1}^\ell)=\sum_{j=1}^{\dim(\mcl{A}')}W_j^*[X_{ij}] W_j,
    \end{align*}
  where  $W_j:= diag(V_{j_1}\otimes I_{n_1},\cdots,V_{j_\ell}\otimes I_{n_\ell})$. Let $m=\sum_{k=1}^\ell m_k$, and define a unital quantum channel $\Psi$ on $\M{m}$, as
  \begin{align*}
      \Psi(Y):=\sum_{j=1}^{\dim(\mcl{A}')} K_j^*YK_j,
  \end{align*}
  where $K_j:=diag(V_{j_1} V_{j_2},\cdots,V_{j_\ell})$.
  Therefore, the Choi-matrix of  $\eta$ (up to permutation equivalence) is equal  to the tensorial Schur product (see \cite{VsSu15}) of $\frac{1}{1-\epsilon} (C_{\Psi}-\epsilon Z)$ and $\ranko{f}{f}$. That is,
    \begin{align*}
        \frac{1}{1-\epsilon} (C_{\Psi}-\epsilon Z)\circ^\otimes \ranko{f}{f},
    \end{align*}
    where $Z=\oplus_{k=1}^\ell I_{m_k^{2}}$ and $f=\sum_{k=1}^\ell \ranko{\Omega_{n_k}}{\Omega_{n_k}}$. Since $C_{\Psi}$ is invertible  and positive, we get $C_{\Psi}+\Delta\geq 0$ for all Hermitian $\Delta$ with $\norm{\Delta}\leq \lambda_{\min}(C_{\Psi})$. Here $\lambda_{\min}(C_{\Psi})$ denotes the smallest eigenvalue of $C_{\Psi}$. Now, by choosing $0<\epsilon_1<\lambda_{\min}(C_{\Psi})$, $C_{\Psi}-\epsilon_1 Z$ is positive and hence $\frac{1}{1-\epsilon_1}(C_{\Psi}-\epsilon_1 Z)\circ^\otimes \ranko{f}{f}$ is positive for all $\epsilon_1\leq \lambda_{\min}(C_{\Psi})$.
    Take $0<\epsilon \leq \min\{\lambda_{\min}(C_{\Psi}),1-\frac{\norm{\Phi-\mbf{E}_{\mcl{A}}}}{r_{\mcl{A}}}\}$, then
    \begin{align*}
        \norm{\eta-\mbf{E}_{\mcl{A}}}=\frac{1}{1-\epsilon}\norm{\Phi-\epsilon\mbf{E}_{\mcl{A}}-\mbf{E}_{\mcl{A}}+\epsilon\mbf{E}_{\mcl{A}}}=
        \frac{1}{1-\epsilon}\norm{\Phi-\mbf{E}_{\mcl{A}}}\leq r_{\mcl{A}}.
    \end{align*}
From Lemma \ref{lem-mix uni-ball}, it follows that $\eta$ is a mixed unitary channel.
Now, by Proposition~\ref{Refined-marcus-ree}, there exists $U \in \mcl{U}(d) \cap \mcl{A}^{'}$ such that $\ip{\eta, \eta} \leq \ip{\mathrm{Ad}_U, \eta}$. Equivalently, we have
    $$
    \frac{1}{1-\epsilon} (\ip{\Phi, \Phi} - 2\epsilon \ip{\Phi, \mbf{E}_{\mcl{A}}} + \epsilon^2 \ip{\mbf{E}_{\mcl{A}},\mbf{E}_{\mcl{A}}}) \leq \ip{\mathrm{Ad}_U, \Phi} - \epsilon \ip{\mathrm{Ad}_U, \mbf{E}_{\mcl{A}}}.
    $$
Using \eqref{exp-ip-channel}, the above equation reduces to  
    $$
    \ip{\mathrm{Ad}_U, \Phi} \geq  \frac{\ip{\Phi, \Phi}  -\epsilon\dim(\mcl{A})}{1-\epsilon} = \ip{\Phi, \Phi} + \frac{\epsilon}{1-\epsilon} (\ip{\Phi, \Phi} -\dim(\mcl{A})).
    $$
Since $\Phi \neq \mbf{E}_{\mcl{A}}$, Lemma \ref{lem-ip-grea} guarantees that $\ip{\Phi,\Phi}>\dim(\mcl{A})$. Consequently, $\ip{\mathrm{Ad}_U, \Phi}  > \ip{\Phi, \Phi}$. This completes the proof.
\end{proof}

\begin{cor}
There exists a neighborhood of $\delta_d$ in which every unital channel $\Phi$, different from $\delta_d$ with Choi-rank $d^2$, is not Erd\H{o}s.
\end{cor} 
\begin{proof}
    Consider the unital $C^*$-subalgebra  $\mcl{A}:=\mbb{C}I_d$ of $\mbb M_d$. Then, $\mbf{E}_{\mcl{A}}=\delta_d$ and $\dim(\mcl{A}') =\dim(\M{d})= d^2$. Thus, if we take any unital channel $\Phi$ that satisfies $0 < \|\Phi - \delta_d \| \leq r_{\mcl{A}}$ with Choi-rank $d^2$, then by Theorem~\ref{nonErdos-suff-cond}, the desired conclusion follows. 
\end{proof}

\section{Erd\H{o}s channels on $\M{d}$}\label{S5}

Inspired by the classical theory of RCDS matrices~\cite{BrDa21}, which is a subclass of doubly stochastic matrices, we introduce an quantum analog of it for unital quantum channels on $\M{d}$, and study its structure. As an application, we classify all the \emph{RCDS} Erd\H{o}s channels on $\M{d}$, and answer Question~\ref{all rank -question}. We further explicitly classify all Erd\H{o}s channels of Choi-rank $2$ and on $\M{2}$.

\subsection{Quantum RCDS channels}
 We generalize the concept of restricted common diagonal sum doubly stochastic matrix to the quantum setup. Recall that, for a matrix $A$, the \emph{skeleton} of $A$, denoted by  $Skel(A)$ is the matrix defined as 
\begin{align*}
    Skel(A)_{ij}=\begin{cases}
        1& \text{ if } A_{ij}\neq 0,\\
        0& \text{ otherwise}.
    \end{cases}
\end{align*} 
A doubly stochastic matrix $A$ is said to be a restricted common diagonal sum (RCDS) matrix if $\min_{P\in \mcl{P}_d(A)}\ip{P,A}=\max_{P\in\mcl{P}_d(A)}\ip{P,A}$, where $\mcl{P}_d(A):=\{P\in\mcl{P}_d: Skel(A)\geq P ~ \text{(entrywise)}\}$. Equivalently, $\ip{P,A}$ is constant, for all $P\in\mcl{P}_d(A)$.
Note that $Skel(A)-P$ is entrywise nonnegative is equivalent to $A-\lambda P$ is entrywise nonnegative, for some $\lambda>0$. Therefore, $A$ is RCDS if and only if $\ip{P,A}$ is constant for all permutation matrix $P$ with $A-\lambda P$ is entrywise nonnegative for some $\lambda>0$ if and only if $\ip{B,A}$ is constant, for all doubly stochastic matrix $B$ such that $A-\lambda B$ is entrywise nonnegative for some $\lambda>0$.  This  observation leads us to extend the notion of RCDS to a non commutative setup as follows.
\begin{defn}\label{quantum-rcds-def}
A unital quantum channel $\Phi$ is said to be \emph{quantum RCDS} if there exists a constant $c>0$ such $\ip{\Psi,\Phi}=c$, for all unital quantum channels $\Psi$ such that $\Phi-\lambda \Psi$ is CP.
\end{defn}

A quantum RCDS mixed unitary channel is uniquely determined by the range of its Choi matrix.
  \begin{prop}
      Let $\Phi_1,\Phi_2$ be two mixed unitary channels with $P_{\Phi_1}=P_{\Phi_2}$. Suppose $\Phi_1,\Phi_2$ are quantum RCDS. Then $\Phi_1=\Phi_2$. 
  \end{prop}
  \begin{proof}
 Since $P_{\Phi_1}=P_{\Phi_2}$, it follows that $\text{Range}(C_{\Phi_1}) = \text{Range}(C_{\Phi_2})$. Suppose $U\in (S_{\Phi_1}\cup S_{\Phi_2})\cap \mcl{U}(d)$. Then, $\text{Range}(C_{\mathrm{Ad}_U})\subseteq \text{Range}(C_{\Phi_1}) = \text{Range}(C_{\Phi_2})$. By Douglas Lemma, for each $j=1,2$, there exists a $\lambda_j>0$ such that $ \lambda_j C_{\Phi_j}-C_{\mathrm{Ad}_U}$ is positive semidefinite. Equivalently, $\Phi_j-\lambda_j^{-1} \mathrm{Ad}_U$ is CP. Since $\Phi_1,\Phi_2$ are quantum RCDS, there exist $c_j>0$ such that $\ip{\mathrm{Ad}_U,\Phi_j}=c_j$ for all unitary $U\in S_{\Phi_1}\cup S_{\Phi_2}$. By Lemma~\ref{lem-mix-unit-int-rep},  we have  the probability measures $\mu_j$ on $S_{\Phi_j}\cap\mcl{U}(d)$ such that $\Phi_j= \int_{S_{\Phi_j} \cap \mathcal U(d)} \mathrm{Ad}_U d\mu_j(U)$, for $j=1,2$. Thus, 
 
 $$
 \ip{\Phi_1, \Phi_2} =  \int_{S_{\Phi_1} \cap \mathcal U(d)} \ip{\mathrm{Ad}_U,  \Phi_2} d\mu_1(U) = c_2  \int_{S_{\Phi_1} \cap \mathcal U(d)} d\mu_1(U) = c_2,
 $$
 and
 $$
 \ip{\Phi_2, \Phi_2} =  \int_{S_{\Phi_2} \cap \mathcal U(d)} \ip{\mathrm{Ad}_U,  \Phi_2} d\mu_2(U) = c_2 \int_{S_{\Phi_2} \cap \mathcal U(d)} d\mu_2(U) = c_2.
 $$
 Similarly, we also obtain $\ip{\Phi_2, \Phi_1} = \ip{\Phi_1, \Phi_1} = c_1$. Consequently,  
 $$
 \ip{\Phi_1-\Phi_2,\Phi_1-\Phi_2} = \ip{\Phi_1, \Phi_1} + \ip{\Phi_2, \Phi_2} - \ip{\Phi_1, \Phi_2}  - \ip{\Phi_2, \Phi_1} = 0. 
 $$
 This shows that $\Phi_1=\Phi_2$.
 \end{proof}

Interestingly, we obtain a structural characterization of quantum RCDS channels on $\M{d}$. Before we state and prove it, we need the following lemma. 
\begin{lem}\label{prop-perp-rep}
 Let $\Phi$ be a unital quantum channel on $\M{d}$. Let $P_{\Phi}$ denote the orthogonal projection onto the range of $C_{\Phi}$. Define
 \begin{itemize}
     \item [(i)] $\msc{H}_{\Phi} = \{W \in \M{d} \otimes \M{d} : W = W^*, ~ P_\Phi W P_\Phi = W \}$,
     \item [(ii)] $V_\Phi = \{ W \in \msc{H}_{\Phi} : (\id \otimes \tr_d)(W) =  (\tr_d \otimes \id)(W) = 0  \}$,
     \item [(iii)] $L_\Phi = \{ P_\Phi(I\otimes H +K \otimes I) P_\Phi :~~ H =H^*, K=K^* \}$.
 \end{itemize}
Then, $V_\Phi^\perp \cap \msc{H}_{\Phi} = L_\Phi$.
\end{lem}
\begin{proof}
We first show that $L_\Phi^\perp  \cap \msc{H}_{\Phi} = V_\Phi$.
To this end, let $W =\sum_{i,j=1}^d E_{ij} \otimes W_{ij} \in \msc{H}_{\Phi}$  and $H, K\in\M{d}$ be Hermitian matrices. Then, 
\begin{align}\label{eq-V-Phi-L-Phi-rela}
\ip{(\tr_d \otimes \id)(W), H} +  \ip{(\id \otimes \tr_d)(W), K} &= \ip{\sum_{i=1}^d W_{ii}, H} + \ip{\sum_{i,j=1}^d E_{ij } \otimes \tr_d(W_{ij}), K}\notag\\
&= \ip{W, I \otimes H}+ \ip{W, K \otimes I}\notag\\
&= \ip{P_\Phi W P_\Phi, I \otimes H + K \otimes I} \notag\\
&= \ip{W, P_\Phi(I \otimes H + K \otimes I)P_\Phi}.
\end{align}
Clearly, if $W\in V_{\Phi}$ then $W\in\msc{H}_{\Phi}\cap L_{\Phi}^{\perp}$. For the other inclusion, by taking   $K=0$ and $H  = (\tr_d \otimes \id)(W)$ in~\eqref{eq-V-Phi-L-Phi-rela} we get  $(\tr_d \otimes \id)(W) = 0$. Similarly, by taking $H=0$ and $K=(\id \otimes \tr_d)(W)$ we also get $(\id \otimes \tr_d)(W) = 0$.

Finally, $\msc{H}_{\Phi} = L_\Phi \oplus V_\Phi$ follows from $L_{\Phi}^\perp \cap\msc{H}_{\Phi}=V_{\Phi}$ together with fact $L_\Phi \subseteq \msc{H}_{\Phi}$. So, if $W \in V_\Phi^\perp \cap \msc{H}_{\Phi}$,  then $W \in L_\Phi.$ Conversely, since  $L_\Phi \subseteq V_\Phi^\perp$, we conclude that $L_\Phi \subseteq V_\Phi^\perp \cap \msc{H}_{\Phi}$. This completes the proof.
\end{proof}

  \begin{thm}\label{thm-Strong qRCDS-struc}
      Let $\Phi$ be a unital quantum channel on $\M{d}$. Then $\Phi$ is quantum RCDS if and only if there exist Hermitian matrices $H,K\in\M{d}$ such that $C_{\Phi}=P_{\Phi}(I\otimes H+K\otimes I)P_{\Phi}$.
  \end{thm}
  \begin{proof}
      Assume that $C_{\Phi}=P_{\Phi}(I\otimes H+K\otimes I)P_{\Phi}$ for some Hermitian matrices $H,K\in\M{d}$. Let $\Psi$ be a unital quantum channel such that $\Phi-\lambda \Psi$ is CP. Equivalently, $C_\Phi - \lambda C_\Psi$ is positive semidefinite. This implies $\text{Range}(C_\Psi) \subseteq \text{Range}(C_\Phi)$, and therefore $P_\Phi C_{\Psi}P_\Phi = C_{\Psi}$.
Thus, 
\begin{align*}
      \ip{\Psi, \Phi} &= \tr(C_\Psi  P_\Phi (I \otimes H +K \otimes I) P_\Phi) = \tr(P_\Phi C_\Psi P_\Phi (I \otimes H +K \otimes I)) 
      =\tr(C_{\Psi}(I\otimes H+K\otimes I)) \notag \\
          &=\sum_{i,j=1}^d\tr(E_{ij})\tr(\Psi(E_{ij})H)+\sum_{i,j=1}^d\tr(E_{ij}K)\tr(\Psi(E_{ij}))  =\tr(H)+\tr(K),
 \end{align*}
  where  the last equality follows from the fact that $\Psi$ is unital and trace preserving.
       Conversely, assume that $\Phi$ is quantum RCDS. By virtue of the Lemma~\ref{prop-perp-rep}, it is enough to prove $C_{\Phi}\in \msc{H}_{\Phi}\cap V_{\Phi}^{\perp}$. Clearly, $C_{\Phi}\in\msc{H}_{\Phi}$. Let $W\in V_{\Phi}$. Since $P_{\Phi}WP_{\Phi}=W$, we have an $\epsilon>0$ such that $C_{\Phi}\pm\epsilon W$ is positive semidefinite. Thus, there exists a CP map $\Psi_{\pm}$ on $\M{d}$ such that $C_{\Psi_{\pm}}=C_{\Phi}\pm\epsilon W$. As $W\in V_{\Phi}$, we get $\tr_j(C_{\Psi_{\pm}})=I$ for $j=1,2$. Consequently, $\Psi_{\pm}$ is a unital quantum channel. Furthermore, $C_{\Phi}-\frac{1}{2}C_{\Psi_{\pm}}=\frac{1}{2}(C_{\Phi}\mp\epsilon W)$ is positive semidefinite. Therefore, $\Phi-\frac{1}{2}\Psi_{\pm}$ is CP, and hence by our assumption $\ip{\Psi_{+},\Phi}=\ip{\Psi_{-},\Phi}$. Consequently, we get $\ip{W,C_{\Phi}}=0$. This shows that $C_\Phi\in V_{\Phi}^{\perp}\cap\msc{H}_{\Phi}$, and completes the proof.  
  \end{proof}

Given $\mathcal B \subseteq \mcl U(d)$, a unital quantum channel $\Phi:\M{d}\to\M{d}$ is called mixed \it{$\mcl B$-unitary} channel if  there exist $\ell \in \mathbb N,$  $\{\lambda_j\}_{j=1}^{\ell}\subseteq [0,1]$  and unitaries $\{U_j\}_{j=1}^\ell\subseteq \mcl B$ such that $\sum_{j=1}^\ell\lambda_j=1$ and
 \begin{align*}
     \Phi(X)=\sum_{j=1}^{\ell}\lambda_j U_j^{*}X U_j, \quad  X\in \M{d}.
 \end{align*}
For an orthogonal set $\mathcal B \subseteq \mcl U(d)$ with $|\mathcal B| = m$ and any integer $k \in \{1, \ldots, m\}$, a unital quantum channel $\Phi$ is  
called \it{$k$-average of $\mathcal B$} if there exists $S \subseteq \mathcal B$ of cardinality $k$ such that 
$$
\Phi = \frac{1}{k}\sum_{W \in S} \mathrm{Ad}_W.
$$
With this terminologies, we provide a family of quantum RCDS channels illustrating Theorem~\ref{thm-Strong qRCDS-struc}.  
\begin{eg}\label{quantumRCDS_example}
Suppose that $\Phi$ is $k$-average of $\mathcal B$ given by $\Phi = \frac{1}{k}\sum_{W \in S} \mathrm{Ad}_W$. Then,
\begin{align*}
C_\Phi &= \frac{1}{k} \sum_{i,j=1}^d E_{ij} \otimes \sum_{W \in S} \mathrm{Ad}_W(E_{ij}) = \frac{1}{k} \sum_{W \in S} \sum_{i,j=1}^d (I \otimes W^*)(E_{ij} \otimes E_{ij}) (I \otimes W)\\
&= \frac{1}{k} \sum_{W \in S} \sum_{i,j=1}^d 
(I \otimes W^*) (\ranko{e_i \otimes e_i}{e_j \otimes e_j}) (I \otimes W)= \frac{1}{k} \sum_{W \in S} \ranko{(I \otimes W^*)e}{(I \otimes W^*)e},
\end{align*}
where $e = \sum_{i=1}^d e_i \otimes e_i$. Further, since $S$ is an orthogonal set, it follows for any $V,W \in S$,  that
$$
\ip{(I \otimes V^*)e, (I \otimes W^*)e} = \sum_{i,j=1}^d \ip{e_i \otimes V^*e_i, e_j \otimes W^*e_j} = \tr(VW^*) = \begin{cases}
    d, & \text{if $V=W$},\\
    0, & \text{otherwise}.
\end{cases} 
$$
This implies that $\{d^{-1/2}(I \otimes W^*)e : W \in S \}$ is an orthonormal basis of $\text{Range}(C_\Phi)$, and hence $P_\Phi = \sum_{W \in S} \ranko{\frac{(I \otimes W^*)e}{\sqrt{d}}}{\frac{(I \otimes W^*)e}{\sqrt{d}}}$. Therefore, if we take $H = K = \frac{d}{2k}I$, then 
$$
P_\Phi (I \otimes H + K \otimes I)P_\Phi = \frac{d}{k} P_\Phi = C_\Phi.
$$
By Theorem~\ref{thm-Strong qRCDS-struc}, $\Phi$ is quantum RCDS.
\end{eg}

   \begin{cor}\label{thm-RCDS-quaRCDS}
      Let $A$ be a doubly stochastic matrix, and let $\Phi_A$ be the channel given as in \eqref{phiA}. Then $\Phi_A$ is quantum RCDS if and only if $A$ is RCDS.
  \end{cor}
  \begin{proof}
Assume that $\Phi_A$ is quantum RCDS. That is, there exists $c > 0$  such that $\ip{\Psi,\Phi_A}=c$ for all unital quantum channel $\Psi$ with $\Phi-\lambda\Psi$ is CP for some $\lambda>0$. Let $P_0$ be any permutation matrix such that $Skel(A)-P_0$ is entrywise nonnegative matrix. By \cite[Lemma 2.4]{PHSK26}, $A=\lambda_0 P_0+\sum_{j=1}^m \lambda_j P_j$ for some permutation matrices $P_j$'s and scalars $\{\lambda_0,\lambda_1,\cdots,\lambda_m\}\subseteq [0,1]$  with $\lambda_0>0$ and $\sum_{j=0}^m\lambda_j =1$. Now, from the proof of Proposition~\ref{prop-doubly-mixed}, it follows that \begin{align}\label{eq-Phi_A-formula}
\Phi_A=\sum_{j=0}^{d-1}\sum_{\ell=0}^m\frac{\lambda_{\ell}}{d}\mathrm{Ad}_{V^jP_{\ell}^*},
    \end{align}
    where $V=diag(1,\omega,\cdots, \omega^{d-1})$ and $\omega=e^{\frac{2\pi i}{d}}$.
Then by \eqref{eq-Phi_A-formula}, $\Phi_A-\frac{\lambda_0}{d}\mathrm{Ad}_{P_0^*}$ is CP and hence by our assumption, we have $\ip{\mathrm{Ad}_{P_0^{*}},\Phi_A}=c$.  But, $\ip{\mathrm{Ad}_{P_0^{*}},\Phi_A}\overset{\eqref{eq-unitary-doublystochastic}}=\ip{B(P_0^{*}),A}$, where $B(P)_{ij}=\abs{P_{ji}}^2$. Therefore, $\ip{P_0,A}=\ip{B(P_0^{*}),A}=c$. As $P_0$ is arbitrary, we conclude that $A$ is a RCDS matrix. 

Conversely, suppose $A$ is a RCDS matrix. Then, by \cite[Theorem~2.6]{BrDa21}, there exist vectors $u=(u_1,\ldots,u_d)$ and $v=(v_1,\ldots,v_d)$ in $\mathbb R^d$ such that 
\begin{equation}
    \label{A-description}
    A = \big[(u_i + v_j) Skel(A)_{ij}\big]_{i,j=1}^d.
\end{equation} 
Since $C_{\Phi_A} = \sum_{i=1}^d E_{ii} \otimes \Phi_A(E_{ii}) = \sum_{i,j=1}^d a_{ji} E_{ii} \otimes E_{jj}$, we get $P_{\Phi_A} = \sum_{i,j=1}^d Skel(A)_{ji} E_{ii} \otimes E_{jj}$. Here $P_{\Phi_A}$ denotes the orthogonal projection onto $\text{Range}(C_{\Phi})$.  
Now, if we take $H: = \text{diag}(u_1, \ldots, u_d)$,
then 
\begin{align*}
P_{\Phi_A}(I\otimes H)P_{\Phi_A} &=   \Big(\sum_{m,n=1}^d Skel(A)_{nm} E_{mm} \otimes E_{nn}\Big) \Big(I\otimes \sum_{i=1}^d u_i E_{ii}\Big) \Big(\sum_{k,l=1}^d Skel(A)_{lk} E_{kk}
\otimes E_{ll} \Big)\\
&= \sum_{m,n,k,l, i=1}^d u_i Skel(A)_{nm} Skel(A)_{lk} (E_{mm} \otimes E_{nn}) (I \otimes E_{ii}) (E_{kk} \otimes E_{ll})\\
&= \sum_{i,j=1}^d u_j Skel(A)_{ji}  E_{ii} \otimes E_{jj}.
\end{align*}
Similarly, by taking $K: = \text{diag}(v_1, \ldots, v_d)$, we obtain $P_{\Phi_A}(K\otimes I)P_{\Phi_A} = \sum_{i,j=1}^d v_i Skel(A)_{ji}  E_{ii} \otimes E_{jj}$.
Therefore, 
\begin{align*}
C_{\Phi_A} = \sum_{i,j=1}^d a_{ji} E_{ii} \otimes E_{jj}
\overset{\eqref{A-description}}= \sum_{i,j=1}^d  (u_j + v_i) Skel(A)_{ji} E_{ii} \otimes E_{jj}
= P_{\Phi_A} (I \otimes H + K \otimes I) P_{\Phi_A}.
\end{align*}
Hence, the desired conclusion follows from the Theorem~\ref{thm-Strong qRCDS-struc}.
\end{proof}

  \begin{thm}\label{thm-RCDS-Erdos}
      Let $\Phi$ be a quantum RCDS mixed unitary channel on $\M{d}$, and let $S_\Phi$ denote the Kraus space of $\Phi$. Then $\Phi$ is Erd\H{o}s if and only if $\ip{\Phi,\Phi}\geq \ip{\mathrm{Ad}_U,\Phi}$ for every unitary $U\notin S_{\Phi}$.
  \end{thm}
  \begin{proof}
      By Theorem~\ref{thm-Strong qRCDS-struc}, there exist Hermitian matrices $H,K$ such that $C_{\Phi}=P_{\Phi}(I\otimes H+K\otimes I)P_{\Phi}$. Note that for any $U \in S_{\Phi}\cap\mcl{U}(d)$, $\mathrm{Ad}_U$ is a unital quantum channel that satisfies $\Phi-\lambda\mathrm{Ad}_U$ is CP for some $\lambda>0$. Therefore, as shown in Theorem~\ref{thm-Strong qRCDS-struc}, we get $\ip{\mathrm{Ad}_U, \Phi} = \tr(H+K)$. Hence
      \begin{align*}
          \max_{U\in \mcl{U}(d)}\ip{\mathrm{Ad}_U,\Phi}&=\max\Big\{\max_{U\in S_{\Phi}\cap\mcl{U}(d)}\ip{\mathrm{Ad}_U,\Phi},\max_{U\notin S_{\Phi}}\ip{\mathrm{Ad}_U,\Phi}\Big\}\\
          &=\max\big\{\tr(H+K),\max_{U\notin S_{\Phi}}\ip{\mathrm{Ad}_U,\Phi}\big\}.
      \end{align*}
      Suppose $\Phi$ is a mixed unitary channel given by $\Phi=\sum_{j=1}^m\lambda_j\mathrm{Ad}_{U_j}$. Then $U_j\in S_{\Phi}\cap\mcl{U}(d)$. Consequently,
      $\ip{\Phi,\Phi}=\sum_{j=1}^m\lambda_j\ip{\mathrm{Ad}_{U_j},\Phi}=\tr(H+K)=\max_{U\in S_{\Phi}\cap\mcl{U}(d)}\ip{\mathrm{Ad}_U,\Phi}$. Thus, $\Phi$ is Erd\H{o}s if and only if $\ip{\Phi,\Phi}=\max_{U\in \mcl{U}(d)}\ip{\mathrm{Ad}_U,\Phi}$ if and only if $\tr(H+K)=\ip{\Phi,\Phi}\geq \ip{\mathrm{Ad}_U,\Phi}$ for all unitary $U\notin S_{\Phi}$.
  \end{proof}

We continue Example~\ref{quantumRCDS_example} and show that the class of channels presented there is also Erd\H{o}s.

\begin{eg} \label{quantumRCDS-example-continue}
Recall from Example~\ref{quantumRCDS_example} that the $k$-average of $\mathcal B$, given by $\Phi = \frac{1}{k}\sum_{W \in S} \mathrm{Ad}_W$,  is quantum RCDS.  
Since the set $S$ is orthogonal, it follows that $S$ is linearly independent, and hence the Choi-rank of $\Phi$ is exactly $k$. 
Moreover, since $S$ is orthogonal with $|S| = k$, we have
\begin{eqnarray} \label{innerproduct-phik}
    \ip{\Phi, \Phi} = \frac{1}{k^2}\sum_{A, B \in S}  \ip{\mathrm{Ad}_A, \mathrm{Ad}_B}
    \overset{\eqref{Ad-Ip-relation}}= \frac{1}{k^2}\sum_{A, B \in S}  |\tr(A^*B)|^2 
    = \frac{d^2}{k^2}\sum_{A \in S} 1 = \frac{d^2}{k}.
\end{eqnarray}
On the other hand, for any unitary $U$, by the Bessel's inequality,
\begin{eqnarray*}
    \ip{\mathrm{Ad}_U, \Phi} = \frac{1}{k}\sum_{W \in S}\ip{\mathrm{Ad}_U, \mathrm{Ad}_W} \overset{\eqref{Ad-Ip-relation}}= \frac{1}{k}\sum_{W \in S}\abs{\tr(U^*W)}^2 
    &\leq& \frac{d}{k} \|U\|_2^2 = \frac{d^2}{k}.
\end{eqnarray*}
This together with \eqref{innerproduct-phik} implies that $\max_{U \in \mcl{U}(d)}\ip{\mathrm{Ad}_U, \Phi} \leq \ip{\Phi, \Phi}$, and hence by Theorem~\ref{thm-RCDS-Erdos}, $\Phi$ is Erd\H{o}s.
\end{eg}
  
As a consequence of Theorem~\ref{thm-RCDS-Erdos}, we also recover the classification result \cite[Lemma~1.6]{PHSK26} related to Erd\H{o}s matrices.  
  
  \begin{cor}\label{classical-Erdos}
      Let $A=[(u_i+v_j)Skel(A)_{ij}]_{i,j=1}^d$ be a doubly stochastic RCDS matrix. Then $A$ is Erd\H{o}s if and only if $$\sum_{\{i,j: P_{ij}=1, a_{ij}=0\}} (u_i+v_j)\geq 0$$ for every permutation matrix $P=[P_{ij}]\notin \mcl{P}_d(A) =\{P\in\mcl{P}_d: Skel(A)\geq P ~ \text{(entrywise)}\}$.
  \end{cor}
  \begin{proof}
      Let $A$ be a doubly stochastic RCDS matrix and let $P\notin \mcl{P}_d(A)$ and $Q\in\mcl{P}_d(A)$  be the permutation matrices given by $Qe_j=e_{\sigma(j)}$ and $Pe_j=e_{\pi(j)}$, for some permutation $\sigma$ and $\pi$. Since $A$ is RCDS, we get $\ip{A,A}=\ip{Q,A}=\sum_{j=1}^d u_j+v_j$ for any $Q\in\mcl{P}_d(A)$. Therefore
      \begin{align*}
          \ip{A-P,A} &=\sum_{j=1}^d u_j+v_j-\sum_{j=1}^d a_{\pi(j)j}
          =\sum_{j=1}^d u_j+v_j-\sum_{j:a_{\pi(j)j}\neq 0}u_{\pi(j)}+v_j\\
          &=\sum_{j:a_{\pi(j)j}= 0}u_{\pi(j)}+v_j
          =\sum_{i,j: P_{ij}=1, a_{ij}=0} u_i+v_j.
      \end{align*}
      Thus, it is enough to show that $A$ is Erd\H{o}s if and only if $\ip{A,A}\geq \ip{P,A}$ for all permutation matrix $P\notin\mcl{P}_d(A)$. 
       By Theorem~\ref{thm-RCDS-quaRCDS} the map $\Phi_A$ given by ~\eqref{phiA} is a quantum RCDS mixed unitary channel. From equations ~\eqref{eq-classical-quantum-ip-relation} and ~\eqref{eq-classical-quantum-ip-1}, we also have $\ip{\Phi_A,\Phi_A}=\ip{A,A}$ and $\ip{\mathrm{Ad}_U,\Phi_A}=\ip{B(U),A}$, where $B(U)_{ij}=\abs{U_{ji}}^2$ and $U$ be any unitary. Now, by Remark~\ref{rmk-Erdos-Erdos channel} and Theorem~\ref{thm-RCDS-Erdos}, it is enough to show $\ip{A,A}\geq \ip{P,A}$ for all permutation matrix $P\notin\mcl{P}_d(A)$  if and only if $\ip{A,A}\geq \ip{B(U),A}$ for every unitary $U\notin S_{\Phi}$. For, first assume that $\ip{A,A}\geq \ip{B(U),A}$ for every unitary $U\notin S_{\Phi}$. Now, let $P\notin\mcl{P}_d(A)$, then we claim that $P^*\notin S_{\Phi_A}$. Suppose not, that is, $P^*\in S_{\Phi_A}$, then  $\text{range}(C_{\mathrm{Ad}_{P^*}})\subseteq \text{Range}(C_{\Phi_A})$, or equivalently, there exists $\lambda>0$ such that $C_{\Phi}-\lambda C_{\mathrm{Ad}_{P^*}}$ is positive semidefinite. In particular, for every $1\leq r,s\leq d$,
 \begin{align*}
     0\leq \ip{e_r\otimes e_s,(C_{\Phi_A}-\lambda C_{\mathrm{Ad}_P^*})e_r\otimes e_s}=a_{sr}-\lambda \abs{P_{sr}}^2= a_{sr}-\lambda P_{sr}.
 \end{align*}
 This shows that $A-\lambda P$ is entrywise nonnegative matrix, and hence $P\in\mcl{P}_d(A)$. This is a contraction and proves the claim. By our assumption, $\ip{A,A}\geq \ip{B(P^*),A}=\ip{P,A}$. Conversely, assume that $\ip{A,A}\geq\ip{P,A}$ for all permutation matrix $P\notin\mcl{P}_d(A)$. Let $U\notin S_{\Phi}$ be unitary, then $B(U)$ is a doubly stochastic matrix, and hence by Birkhoff Theorem $B(U)=\sum_{j=1}^m\lambda_jP_j$ for some permutations $P_j$ and scalrs $\lambda_j\in(0,1)$ with $\sum_{j=1}^m\lambda_j=1$.  Since $A$ is RCDS, $\ip{A,A}=\ip{P_j,A}$ for all $P_j\in\mcl{P}_d(A)$. So, 
 \begin{align*}
    \ip{A,A}-\ip{B(U),A}=\sum_{j=1}^m\lambda_j(\ip{A,A}-\ip{P_j,A})=\sum_{\{j: P_j\notin \mcl{P}_d(A)\}}\lambda_j(\ip{A,A}-\ip{P_j,A})\geq 0
 \end{align*}
 This completes the proof.
  \end{proof}

Towards a solution to Question~\ref{all rank -question}, we first obtain the following result, which classifies Erd\H{o}s channels among all the mixed unitary channels having a representation for which the set of Kraus operators forms an orthogonal set. 
\begin{thm}
\label{Erdos-d}
Let $\mathcal B$ be an orthogonal set of unitary elements in $\M{d}$ with $|\mathcal B| = m$. Then, any $k$-average of $\mathcal B$ is an Erd\H{o}s channel. Conversely, 
suppose that $\Phi$ is a mixed $\mathcal B$-unitary channel on $\M{d}$. If $\Phi$ is Erd\H{o}s, then 
$\Phi$ is $k$-average of $\mathcal B$ for some $1 \leq k \leq m$. In such a case, $k$ is uniquely determined and equal to the Choi-rank of $\Phi$.
\end{thm}
\begin{proof}
Suppose that $\Phi$ is $k$-average of $\mathcal B$. Then, the Erd\H{o}s property follows from Example~\ref{quantumRCDS-example-continue}. 
Conversely, suppose that $\Phi$ is  a mixed $\mathcal B$-unitary channel and is Erd\H{o}s.
    Since $\Phi$ is mixed $\mathcal B$-unitary, we may enumerate the basis $\mathcal B$ by $\{U_j\}_{j=1}^{m}$ so that 
    \begin{equation*}
     \Phi= \sum_{j = 1}^{m} \lambda_j \mathrm{Ad}_{U_j}
    \end{equation*}
    with $\{\lambda_j\}_{j=1}^{m}$ satisfying
\begin{equation}\label{lambda-conditions-d}
\lambda_1 \geq \lambda_2 \geq \ldots \geq \lambda_{m} \geq 0~\text{such that}~\sum_{j=1}^{m} \lambda_j = 1.    
\end{equation}
Since $\{U_i\}_{i=1}^m$ is orthogonal, we have  
\begin{equation*} 
\ip{\Phi, \Phi} = \sum_{i,j=1}^{m} \lambda_i \lambda_j \ip{\mathrm{Ad}_{U_i},\mathrm{Ad}_{U_j}} \overset{\eqref{Ad-Ip-relation}}=\sum_{i,j=1}^{m} \lambda_i \lambda_j \abs{\tr\big(U_i^*U_j\big)}^2 = d^2\sum_{i=1}^{m} \lambda_i^2. 
\end{equation*}
On the other hand, observe that  
$$
\ip{\mathrm{Ad}_{U_1}, \Phi} = \sum_{j=1}^{m}\lambda_j \ip{\mathrm{Ad}_{U_1}, \mathrm{Ad}_{U_j}} = \sum_{j=1}^{m}\lambda_j \abs{\tr\big(U_1^*U_j\big)}^2=\lambda_1 d^2.
$$ 
Therefore, 
since $\Phi$ is Erd\H{o}s, by Lemma~\ref{erdos-kraus-attain}, it follows that
$\lambda_1 = \sum_{j=1}^{m}\lambda_j^2.$
Observe that $\lambda_1$ is strictly positive. Thus, if we let $a_j = \frac{\lambda_j}{\lambda_1}$ for $j =1, \ldots, d^2$, then the above condition yields that 
\begin{equation}\label{cond-l^2-d}
    \sum_{j=2}^{m} a_j^2 = \frac{1}{\lambda_1} - 1.
\end{equation}
Furthermore, by \eqref{lambda-conditions-d}, 
$$\sum_{j=2}^{m} a_j = \frac{1}{\lambda_1} - 1~~\text{and}~~0 \leq a_3 \leq a_2 \leq a_1 \leq 1,$$ 
which combined with \eqref{cond-l^2-d} yields that 
$\sum_{j=2}^{m} a_j^2 = \sum_{j=2}^{m} a_j.$ 
Thus, since $a_j^2 \leq a_j$ for all $j =1, \ldots, m$, we finally conclude that $a_j \in \{0,1\},~j =2, \ldots, m$, or equivalently, $\lambda_j \in \{0, \lambda_1\},~j =2, \ldots, m$. Therefore, if the Choi-rank of $\Phi$ is $k$, then $\lambda_j = \lambda_1$ for all $j=1, \ldots, k$, and $\lambda_{k+1} =0$ (which implies that remaining $\lambda_j$'s are also zero). This together with \eqref{lambda-conditions-d} yields that $\lambda_j = \frac{1}{k}$ for all $j =1, \ldots, k$.
Equivalently, $\Phi$ takes the form 
$
\Phi= \frac{1}{k}\sum_{j = 1}^{k} \mathrm{Ad}_{U_j}.
$
This completes the proof.
\end{proof}

Let $\omega$ be the $d$-th root of unity and let $U \in \M{d}$ be defined by $Ue_i = e_{i+1}$ for all $i=1, \ldots, d-1$ and $Ue_d=e_1$, and let $V = diag(1, \omega, \ldots, \omega^{d-1})$. Note that both $U$ and $V$ are unitary matrices and $U^d = V^d =I_d$.  Then, the collection $\{d^{-1/2}U^mV^n : m,n = 0,1, \ldots, d-1\}$ is an orthonormal basis of $\mbb M_{d}$ consisting of unitaries. This is also known as \it{Weyl-unitary basis}.
As an application of Theorem~\ref{Erdos-d}, we provide an affirmative answer to Question~\ref{all rank -question}.
\begin{cor}\label{Erdos-d-weyl}
Let $d \geq 2$, and let $k$ be an integer such that $1 \leq k \leq d^2$. The following are valid:
\begin{itemize}
    \item [(i)] there exists an Erd\H{o}s Weyl-unitary channel whose Choi-rank is equal to $k$,
    \item [(ii)] any Erd\H{o}s Weyl-unitary channel of Choi-rank $k$ (up to unitary equivalence) is given by 
    $$\frac{1}{k}\sum_{W \in S} \mathrm{Ad}_W$$ for some set $S$ of Weyl-unitaries with $|S| = k$.
\end{itemize}
\end{cor}
\begin{proof}
Let $k$ be an integer such that $1 \leq k \leq d^2$, and let $\mathcal B = \{d^{-1/2}U^mV^n : m,n = 0,1, \ldots, d-1\}$ be the Weyl-unitary basis of $\M{d}$. To see (i), consider a subset $S \subseteq \mathcal B$ of cardinality $k$. Since $\mcl B$ is an orthonormal basis of $\mbb M_{d}$ consisting of unitaries, if we define
$$\Phi_{k, S} := \frac{1}{k}\sum_{W \in S} \mathrm{Ad}_W,$$
 then $\Phi_{k, S}$ is $k$-average of $\mathcal B$. Hence, by Theorem~\ref{Erdos-d}, $\Phi_{k, S}$ is Erd\H{o}s with its Choi-rank equal to $k$. This concludes the part (i). The part (ii) is immediate from Theorem~\ref{Erdos-d}. 
\end{proof}

\subsection{Erd\H{o}s channels of Choi-rank $2$} In this subsection, we classify all the Erd\H{o}s channels of Choi-rank $2$. We begin by presenting the following result, which establishes a necessary condition for a unital quantum channel to be Erd\H{o}s. This result can be viewed as a quantum analog of \cite[Proposition~1.2]{Rag25}.

\begin{lem}\label{lem-Erd-scalar-condition}
    Let $\Phi$ be a mixed unitary channel given by $\Phi=\sum_{j=1}^\ell \lambda_j \mathrm{Ad}_{U_j}$, where $\ell$ is the mixed unitary rank of $\Phi$. Let $A=\big[\abs{\tr(U_i^*U_j)}^2\big]_{i,j=1}^\ell$. Then, $A$ is invertible. Moreover, if $\Phi$ is Erd\H{o}s, then $(\lambda_1,\lambda_2,\cdots,\lambda_\ell)^t=\frac{A^{-1}\mbf{1}}{\ip{\mbf{1},A^{-1}\mbf{1}}}$, where $\mathbf{1}$ is the column vector with all entries equal to $1$.
\end{lem}
\begin{proof}
Observe that $A$ is a real Gram matrix. Thus, to check the invertibility of $A$, it is enough to show that $\{\mathrm{Ad}_{U_j}\}_{j=1}^\ell$ is a linearly independent set over $\mathbb R$. Suppose the set $\{\mathrm{Ad}_{U_j}\}_{j=1}^\ell$ is linearly dependent, i.e., there exists a non-zero vector $(\alpha_1,\alpha_2,\ldots,\alpha_\ell)\in\mbb{R}^{\ell}$ such that $\sum_{j=1}^\ell \alpha_j \mathrm{Ad}_{U_j} = 0$. Evaluating this at the identity matrix gives that
\begin{equation}
\label{sum-coeff-zero}
\sum_{j=1}^\ell \alpha_j = 0.
\end{equation}
Since $(\alpha_1,\ldots,\alpha_\ell)$ is a non-zero vector, $I: = \{j : \alpha_j < 0\} \neq \emptyset$. Let $r = \min\{\frac{-\lambda_j}{\alpha_j} : j \in I\}$.
Then, $\lambda_j + r \alpha_j \geq 0$ for all $j$ and there exists $j_0 \in \{1, \ldots, \ell\}$ such that $\lambda_{j_0} + r \alpha_{j_0} = 0$. Hence, by \eqref{sum-coeff-zero},  
$$
\sum_{\substack{j=1 \\ j \neq j_0}}^\ell \lambda_j + r \alpha_j =1 \quad \text{and} \quad \Phi = \sum_{j=1}^\ell \lambda_j \mathrm{Ad}_{U_j} = \sum_{\substack{j=1 \\ j \neq j_0}}^\ell (\lambda_j + r \alpha_j) \mathrm{Ad}_{U_j}, 
$$
which contradicts the fact that the mixed unitary rank of $\Phi$ is $\ell$. This shows that $A$ is invertible.

Now, suppose that $\Phi$ is an Erd\H{o}s channel. Then, by Lemma~\ref{erdos-kraus-attain},
$
\ip{\Phi, \Phi} = \ip{\mathrm{Ad}_{U_j}, \Phi},~ j = 1, \ldots, \ell.
$ Thus, by Lemma~\ref{U-W-IP},
 $$
 A \begin{bmatrix}
    \lambda_1 \\
    \lambda_2 \\
    \vdots \\
    \lambda_\ell
\end{bmatrix} =    \Big[\ip{\mathrm{Ad}_{U_i}, \mathrm{Ad}_{U_j}}\Big]_{i,j=1}^\ell\begin{bmatrix}
    \lambda_1 \\
    \lambda_2 \\
    \vdots \\
    \lambda_\ell
\end{bmatrix} = \begin{bmatrix}
    \ip{\mathrm{Ad}_{U_1},\Phi} \\
    \ip{\mathrm{Ad}_{U_2},\Phi} \\
    \vdots \\
    \ip{\mathrm{Ad}_{U_\ell},\Phi} 
\end{bmatrix} = \ip{\Phi, \Phi}  \begin{bmatrix}
    1 \\
    1 \\
    \vdots \\
    1
\end{bmatrix}.
 $$
 Since $A$ is invertible, it then follows that $(\lambda_1, \ldots, \lambda_\ell)^t = \ip{\Phi, \Phi}A^{-1}\mathbf{1}$, and hence $$\ip{\Phi,\Phi}\ip{\mbf{1},A^{-1}\mbf{1}} = \sum_{i=1}^\ell \lambda_i  =1.$$ 
Therefore,  $(\lambda_1, \ldots, \lambda_\ell)^t = \frac{A^{-1}\mbf{1}}{\ip{\mbf{1},A^{-1}\mbf{1}}}$. This completes the proof.
\end{proof}

We will now demonstrate that the necessary condition stated above for a mixed unitary channel to be Erd\H{o}s is not generally sufficient. For instance, let $d =3$ and let $\omega$ be the cube root of unity. Suppose $\Phi = \lambda_1 \id + \lambda_2 \mathrm{Ad}_{U_2} + \lambda_3 \mathrm{Ad}_{U_3}$, where $U_2 = diag(1,1,-1)$ and $U_3 = diag(1, \omega, \omega^2)$. Clearly, the mixed unitary rank of $\Phi$ is at most $3$. Since $\{I_3, U_2, U_3\}$ is a linearly independent set, the Choi-rank of $\Phi$ is $3$. Because, the mixed unitary rank of $\Phi$ always bounded below by Choi-rank, the mixed unitary rank of $\Phi$ is equal to $3$. Also, note that the matrix $A = \big[|\tr(U_i^*U_j)|^2\big]_{i,j=1}^3$ is given by   
$$
A= \begin{bmatrix}
9 & 1 & 0 \\
1 & 9 & 4 \\
0 & 4 & 9
\end{bmatrix},
$$
which implies that $\frac{A^{-1}\mbf{1}}{\ip{\mbf{1},A^{-1}\mbf{1}}} = 
(\frac{5}{12}, \frac{1}{4}, \frac{1}{3}).
$ 
However, consider the channel $\Phi = \frac{5}{12}\id + \frac{1}{4}\mathrm{Ad}_{U_2} + \frac{1}{3} \mathrm{Ad}_{U_3}$. Then, $\ip{\Phi, \Phi} = 4$, while for $W = diag(1, -\omega^2, \omega^2)$, we get
$
\ip{\mathrm{Ad}_{W}, \Phi} = \frac{54}{12} > \ip{\Phi, \Phi}.
$
Hence, $\Phi$ fails to be Erd\H{o}s.

\begin{thm}[Characterization of Erd\H{o}s channels of Choi-rank $2$]\label{Rank2-classi}
   Let $\Phi$ be a mixed unitary channel with Choi-rank equal to $2$. Then, the following are equivalent:
   \begin{enumerate}[label=(\roman*)]
       \item $\Phi$ is Erd\H{o}s;
       \item (up to unitary equivalence) there exists unitary $U$ with either $\tr(U)=0$ or $U=zU^*$, for some $\abs{z}=1$ such that
       $\Phi=\frac{1}{2}(\id+\mathrm{Ad}_U)$.
   \end{enumerate}
\end{thm}
\begin{proof} Note that any unital quantum channel $\Phi$ of Choi-rank $2$, (up to unitary equivalence) can be written as
$$
\Phi = \lambda \id + (1-\lambda)\mathrm{Ad}_U,
$$
for some $\lambda\in (0,1)$ and $U \in \mcl U(d)$ such that the set $\{I,U\}$ is linearly independent.

$(i) \implies (ii)$ Assume that $\Phi$ is Erd\H{o}s. Since $\{I,U\}$ is a linearly independent set, it follows that mixed unitary rank of $\Phi$ is $2$.
Thus, by Lemma~\ref{lem-Erd-scalar-condition}, $$\frac{A^{-1}\mbf{1}}{\ip{\mbf{1},A^{-1}\mbf{1}}}=\begin{bmatrix}
    \lambda\\1-\lambda
\end{bmatrix},$$
where $A=\Matrix{d^2&\abs{\tr(U)}^2\\ \abs{\tr(U)}^2&d^2}$. Since $A^{-1}\mbf{1}= \frac{1}{d^2+\abs{\tr(U)}^2} \mbf{1}$, we then obtain that 
$$\frac{A^{-1}\mbf{1}}{\ip{\mbf{1},A^{-1}\mbf{1}}}= \frac{\mbf{1}}{\ip{\mbf{1}, \mbf{1}}} = \frac{1}{2}\mbf{1}.$$ This shows that $\lambda=\frac{1}{2}$. 

Now, to get the require conditions on unitary $U$, define the map $f: \mcl U(d) \rightarrow  \mathbb R$  by
$$f(W) =|\tr(W)|^2 + |\tr(W^*U)|^2,\quad W \in \mcl U(d).$$
Since $\Phi$ is Erd\H{o}s, by Lemma~\ref{U-W-IP}, 
$$
\frac{1}{2}\max_{W \in \mcl U(d)} f(W) =  \max_{W \in \mcl U(d)} \ip{\mathrm{Ad}_W, \Phi} = \ip{\Phi, \Phi} = \frac{1}{2} (d^2 + |\tr(U)|^2), 
$$
and hence, $\max_{W \in \mcl U(d)} f(W) = d^2 + |\tr(U)|^2 = f(I)$. This implies that the function $f$ has the global maximum at the identity matrix $I$. In particular, we have 
\begin{equation}
\label{derivative-zero}
   \lim_{t \rightarrow 0} \frac{d}{dt}f(e^{itH}) = 0 \quad \text{for all Hermitian matrices $H$.}
\end{equation}
For each Hermitian matrix $H$, since $f(e^{itH}) = \tr(e^{itH})\tr(e^{-itH})+ \tr(e^{-itH}U) \tr(e^{itH}U^*)$, we get
\begin{eqnarray*}
\frac{d}{dt}f(e^{itH}) &=&  \tr(e^{itH})\tr(-ie^{-itH}H) + \tr(e^{-itH})\tr(ie^{itH}H)+ \tr(e^{-itH}U) \tr(ie^{itH}U^*H) \notag\\
&&  + \tr(e^{itH}U^*) \tr(-ie^{-itH}UH).    
\end{eqnarray*}
Thus, we obtain
$$
\lim_{t \rightarrow 0} \frac{d}{dt}f(e^{itH}) = \tr(ZH), 
$$
where $Z= i \big(\tr(U)U^* - \overline{\tr(U)}U\big)$. Therefore, by \eqref{derivative-zero}, $\tr(ZH\big) = 0$. Since $Z$ is a Hermitian matrix and $H$ is arbitrary, it follows that $Z =0$. Consequently, $\tr(U)U^* = \overline{\tr(U)}U$. This concludes that either $\tr(U) = 0$, or $U = zU^*$, where $z= \big(\overline{\tr(U)}\big)^{-1} \tr(U){}$. 

$(ii) \implies (i)$  Suppose first that $\tr(U) = 0$. Then, $\{\frac{I_d}{\sqrt{d}},\frac{U}{\sqrt{d}}\}$ forms an orthonormal subset of $\M{d}$ with respect to the Hilbert-Schmidt inner product. By Lemma~\ref{U-W-IP} and the Bessel's inequality, 
\begin{align*}
     \ip{\mathrm{Ad}_W,\Phi}=\frac{1}{2}(\abs{\tr(W^*)}^2+\abs{\tr(W^*U)}^2)\leq \frac{d}{2}\norm{W}_2^2=\frac{d^2}{2}=\ip{\Phi,\Phi}, \quad W \in \mathcal U(d).
\end{align*}
This shows that $\max_{W\in\mcl{U}(d)}\ip{\mathrm{Ad}_W,\Phi}\leq \ip{\Phi,\Phi}.$  Now, the equality follows by using Theorem~\ref{Marcus-Ree-noncom}, implying that $\Phi$ is an Erd\H{o}s channel.

For the remaining case, let $U$ be an unitary such that $U=zU^*$ for some $\abs{z}=1$. Then, by Lemma~\ref{U-W-IP}, 
\begin{eqnarray} \label{average-2-new}
    \ip{\Phi,\Phi}=\ip{\frac{\id+\mathrm{Ad}_U}{2},\frac{\id+\mathrm{Ad}_U}{2}}=\frac{1}{2}(d^2+\abs{\tr(U)}^2).
\end{eqnarray}
Since $U = z U^*$ with $|z| = 1$, the unitary matrix $V := \sqrt{z} U$ satisfies $V^2 = I$. It also follows that the matrices $P_{\pm} := \frac{I_d \pm V}{2}$ are projections, whose ranks are given by
\begin{align}\label{rank-projections}
\mathrm{rank}(P_{\pm}) = \frac{d \pm \tr(V)}{2}.
\end{align}
Furthermore, for any $W \in \mcl{U}(d)$, again by Lemma~\ref{U-W-IP},
\begin{eqnarray*}
  \ip{\mathrm{Ad}_W, \Phi} 
    = \frac{|\tr(W)|^2 + |\tr(W^*U)|^2}{2}
    &=& \frac{|\tr(W)|^2 + |\tr(W^*V)|^2}{2}\\
    &=& \frac{|\tr(WP_+) + \tr(WP_-)|^2 + |\tr(WP_+) - \tr(WP_-)|^2}{2}\\
    &=&|\tr(WP_+)|^2 + |\tr(WP_-)|^2\\
    &\leq& \mathrm{rank}(P_+)^2 +\mathrm{rank}(P_-)^2\\
    &\overset{\eqref{rank-projections}}=& \frac{d^2 + \tr(V)^2}{2}=\frac{d^2 + |\tr(U)|^2}{2}.
\end{eqnarray*}
Therefore, by \eqref{average-2-new}, $\max_{W \in \mcl{U}(d)} \ip{\mathrm{Ad}_W, \Phi} \leq \ip{\Phi, \Phi}$. Now, the equality follows by using Theorem~\ref{Marcus-Ree-noncom}, implying that $\Phi$ is an Erd\H{o}s channel.
\end{proof}
In sharp contrast with the classical setting, we show below that Erd\H{o}s channels of Choi-rank $2$ admit uncountably many inequivalent classes.
\begin{rmk}
\label{inequivalent-channels-rank2}
Let $d=4$ and $U_j ~(j=1,2)$ be unitaries on $\M{4}$ given by
$$
U_j = \operatorname{diag}\!\left(1, e^{i\theta_j}, -1, -e^{i\theta_j}\right), \quad j=1,2,
$$
where $\theta_1 \neq \theta_2$ and $\theta_j \in (0, 2\pi) \setminus \{\pi\}$. Note that $\tr(U_j) = 0,~ j=1,2.$ Additionally, assume that
\begin{equation}\label{addi-assump}
\text{$\theta_1 \pm \theta_2\neq k\pi$  for any integer $k$}.
\end{equation}
We claim that the channels $\frac{1}{2}(\id+\mathrm{Ad}_{U_1})$ and $\frac{1}{2}(\id+\mathrm{Ad}_{U_2})$ are not unitarily equivalent. Suppose on the contrary that the maps $\frac{1}{2}(\id+\mathrm{Ad}_{U_1})$ and $\frac{1}{2}(\id+\mathrm{Ad}_{U_2})$ are unitarily equivalent.
Then, there exist unitaries $W_1$ and $W_2$ such that 
$
    \id+\mathrm{Ad}_{U_1}=\mathrm{Ad}_{W_2W_1}+\mathrm{Ad}_{W_2U_2W_1}.
$
Consequently, there exists $\begin{pmatrix}
\alpha & \beta\\
\gamma & \delta
\end{pmatrix} \in \mcl U(2)$ such that 
$$
\binom{W_2W_1}{W_2U_2W_1} = \begin{pmatrix}
\alpha & \beta\\
\gamma & \delta
\end{pmatrix}
\binom{I}{U_1}.
$$
This implies that $W_2W_1=\alpha+\beta U_1$, and therefore, the spectrum  $\sigma(W_2W_1)$ of $W_2W_1$ is $$\sigma(W_2W_1) =
\left\{
\alpha\pm \beta,
\alpha\pm\beta e^{i\theta_1}
\right\}.
$$
Since $W_2W_1$ is unitary, 
\begin{align*}
|\alpha+\beta|^2= |\alpha-\beta|^2 = 1 \quad \text{and} \quad  
|\alpha+\beta e^{i\theta_1}|^2 = |\alpha - \beta e^{i\theta_1}|^2=1,
\end{align*}
and hence, $\text{Re}(\alpha\overline{\beta}) = \text{Re}(\alpha\overline{\beta} e^{-i\theta_1}) = 0$. Therefore, $\alpha\overline{\beta} = iy$ for some $y \in \mathbb R$, which implies that $y \sin \theta_1 = \text{Re}(\alpha\overline{\beta} e^{-i\theta_1}) = 0$. Since $\theta_1 \in (0, \pi) \cup (\pi, 2\pi)$, it then follows that $y=0$, and hence 
$\alpha\overline{\beta}=0.$ 
Since $\begin{pmatrix}
\alpha & \beta\\
\gamma & \delta
\end{pmatrix} \in \mathcal U(2)$, it follows that either
$
(|\alpha|, |\beta|)=(1, 0)$,
or
$
(|\alpha|, |\beta|)=(0, 1)
$ 
and consequently,
\begin{align*}
\binom{W_2W_1}{W_2U_2W_1}~\text{is either}~ \binom{\alpha I}{\delta U_1},
~\text{or}~
\binom{\beta U_1}{\gamma I}.
\end{align*} 
Now, let us treat the case that $\binom{W_2W_1}{W_2U_2W_1} =  \binom{\alpha I}{\delta U_1}$ (the other case can be proved exactly in a similar way). Then, we have
$$
W_1^* U_2W_1 = (W_2W_1)^* (W_2U_2W_1) = \overline{\alpha} \delta U_1.
$$
As $\abs{\alpha}=\abs{\delta}=1$, we get that $W_1^* U_2W_1 = e^{i \theta} U_1$ for some $\theta \in [0, 2\pi)$, and hence, $\sigma(U_2) = \sigma(e^{i \theta} U_1).$ This combined with the fact 
$\sigma(U_2) = \{\pm 1, \pm e^{i\theta_2}\}$, implies that 
$$\{\pm 1, \pm e^{i\theta_2}\} = \{
\pm e^{i\theta},
\,\pm e^{i(\theta_1+\theta)}
\}.
$$
In case, if $e^{i(\theta_1+\theta)} = 1$ and $e^{i \theta} = e^{i \theta_2}$, then both $\theta_1+\theta$ and $\theta - \theta_2$ are an integer multiple of $2\pi$. This yields that $\theta_1+\theta_2 =  2 k \pi$ for some $k \in \mathbb Z$, which contradicts the assumption \eqref{addi-assump}. By using \eqref{addi-assump}, other cases also lead to a contradiction. Thus, we conclude that  $\frac{1}{2}(\id+\mathrm{Ad}_{U_1})$ and $\frac{1}{2}(\id+\mathrm{Ad}_{U_2})$ are inequivalent channels. 
\end{rmk}

\subsection{Erd\H{o}s channels on $\M{2}$} In this subsection, we classify all the Erd\H{o}s channels (up to unitary equivalence) on $\M{2}$. Before we state the classification result, we set some notations. Let $\sigma_1, \sigma_2$ and $\sigma_3$ denote the matrices 
$\begin{pmatrix}
    0 & 1\\
    1 & 0
\end{pmatrix}, \begin{pmatrix}
    0 & -i\\
    i & 0
\end{pmatrix}$ and  $\begin{pmatrix}
    1 & 0\\
    0 & -1
\end{pmatrix}$, respectively. These matrices along with the identity matrix are known as Pauli matrices. Note that the Pauli matrices are unitaries, and form an orthogonal basis of $\M{2}$. In fact, the set 
$$\Big\{\frac{I_2}{\sqrt{2}},\frac{\sigma_j}{\sqrt{2}}: 1\leq j\leq 3\Big\}$$ is an orthonormal basis for $\M{2}$ with respect to the Hilbert-Schmidt inner product.
Below is the classification result on $\M{2}$.

\begin{thm}\label{Erdos-2}
Let $\Phi: \mathbb{M}_2 \to \mathbb{M}_2$ be a unital quantum channel, and consider the unitaries $\sigma_j$ ~($j=1,2,3$) defined above. 
Then, $\Phi$ is an Erd\H{o}s channel if and only if $\Phi$, up to unitary equivalence, belongs to the set
$$
\Big\{\id, \frac{\id +\mathrm{Ad}_{\sigma_1}}{2}, \frac{\id + \mathrm{Ad}_{\sigma_1}+ \mathrm{Ad}_{\sigma_2}}{3}, \frac{\id + \sum_{j=1}^3 \mathrm{Ad}_{\sigma_j}}{4}\Big\}.
$$
\end{thm}
\begin{proof}
     Since $\Phi$ is a unital quantum channel. By \cite[Theorem 2.2]{ChLi23}, there exist $U,W\in\mathcal U(2)$ such that 
    \begin{equation*}
        \mathrm{Ad}_U\circ\Phi\circ\mathrm{Ad}_W(X)= \lambda_0 X+ \lambda_1 \sigma_1 X\sigma_1 + \lambda_2 \sigma_2 X \sigma_2 + \lambda_3 \sigma_3X\sigma_3, \quad X \in \M{2},
    \end{equation*}
    where $\{\sigma_1, \sigma_2,\sigma_3\}$ are Pauli matrices, and $\lambda_0 \geq \lambda_1 \geq \lambda_2 \geq \lambda_3 \geq 0~\text{such that}~\sum_{j=0}^3 \lambda_j = 1.$
Suppose that $\Phi$ is Erd\H{o}s. By \eqref{Erdos-preserve-equivalence}, $\mathrm{Ad}_U\circ\Phi\circ\mathrm{Ad}_W$ is also Erd\H{o}s. As the set $\{I_2, \sigma_1, \sigma_2, \sigma_3\}$ is an orthogonal basis of $\M{2}$ consisting of unitaries, Theorem~\ref{Erdos-d} yields that
$$
\mathrm{Ad}_U\circ\Phi\circ\mathrm{Ad}_W \in \Big\{\id, \frac{\id +\mathrm{Ad}_{\sigma_1}}{2}, \frac{\id + \mathrm{Ad}_{\sigma_1}+ \mathrm{Ad}_{\sigma_2}}{3}, \frac{\id + \sum_{j=1}^3 \mathrm{Ad}_{\sigma_j}}{4}\Big\},
$$
completing the forward part. The converse part is immediate from Theorem~\ref{Erdos-d}. This completes the proof.
\end{proof}

\section*{Acknowledgement}
Both authors acknowledge the support received from the Indian Institute of Technology Bombay through an Institute postdoctoral fellowship. They also thank B. V. R. Bhat for carefully reading the manuscript and for helpful suggestions that improved Proposition~\ref{prop-Delta-range} and the overall presentation. The authors are also thankful to R. Tripathi for several useful suggestions and comments.

 \section*{Data availability}
The manuscript has no associated data.
\section*{Declarations}
 \subsection*{Conflict of interest} The authors have no conflict of interest.
\printbibliography

 \end{document}